\documentclass[sigconf]{acmart}
\pdfoutput=1
\usepackage{setspace}
\usepackage{multirow}

\newtheorem{definition}{Definition}
\usepackage{epstopdf}
\usepackage{enumitem}
\usepackage{mathrsfs}
\usepackage{setspace}
\usepackage{epstopdf}
\usepackage{threeparttable}
\usepackage{soul}
\usepackage[ruled,linesnumbered]{algorithm2e}
\usepackage{enumitem}
\usepackage[title]{appendix}
\AtBeginDocument{%
  }

\copyrightyear{2025}
\acmYear{2025}
\acmConference[KDD '25]{Proceedings of the 31st ACM SIGKDD Conference on Knowledge Discovery and Data Mining V.1}{August 3--7, 2025}{Toronto, ON, Canada}
\acmBooktitle{Proceedings of the 31st ACM SIGKDD Conference on Knowledge Discovery and Data Mining V.1 (KDD '25), August 3--7, 2025, Toronto, ON, Canada}
\acmDOI{10.1145/3690624.3709334}
\acmISBN{979-8-4007-1245-6/25/08}

\begin{document}
\title{PrivDPR: Synthetic Graph Publishing with Deep PageRank under Differential Privacy}

\author{Sen Zhang}
\affiliation{%
 \institution{The Hong Kong Polytechnic University}
 % \streetaddress{1 Th{\o}rv{\"a}ld Circle}
 \city{Hong Kong}
 \country{China}}
\email{senzhang@polyu.edu.hk}

\author{Haibo Hu}
\affiliation{%
 \institution{The Hong Kong Polytechnic University}
 % \streetaddress{30 Shuangqing Rd}
 \city{Hong Kong}
 \country{China}}
\email{haibo.hu@polyu.edu.hk}

\author{Qingqing Ye}
\authornote{Corresponding author}
\affiliation{%
 \institution{The Hong Kong Polytechnic University}
 % \streetaddress{30 Shuangqing Rd}
 \city{Hong Kong}
 \country{China}}
\email{qqing.ye@polyu.edu.hk}

\author{Jianliang Xu}
\affiliation{%
 \institution{Hong Kong Baptist University}
 \city{Hong Kong}
 \country{China}}
 \email{xujl@comp.hkbu.edu.hk}

\begin{abstract}
The objective of privacy-preserving synthetic graph publishing is to safeguard individuals' privacy while retaining the utility of original data. Most existing methods focus on graph neural networks under differential privacy (DP), and yet two fundamental problems in generating synthetic graphs remain open. First, the current research often encounters high sensitivity due to the intricate relationships between nodes in a graph. Second, DP is usually achieved through  advanced composition mechanisms that tend to converge prematurely when working with a small privacy budget. In this paper, inspired by the simplicity, effectiveness, and ease of analysis of PageRank, we design PrivDPR, a novel privacy-preserving deep PageRank for graph synthesis. In particular, we achieve DP by adding noise to the gradient for a specific weight during learning. 
Utilizing weight normalization as a bridge, we theoretically reveal that increasing the number of layers in PrivDPR can effectively mitigate the high sensitivity and privacy budget splitting. Through formal privacy analysis, we prove that the synthetic graph generated by PrivDPR satisfies node-level DP. Experiments on real-world graph datasets show that PrivDPR preserves high data utility across multiple graph structural properties.
\end{abstract} 

% We theoretically demonstrate that we can preset a desired small sensitivity and achieve it by slightly increasing the number of layers.



%\keywords{Do, Not, Us, This, Code, Put, the, Correct, Terms, for,
%  Your, Paper}


\maketitle

\section{Introduction}\label{Intro}
Numerous real-world applications, such as social networks~\cite{leskovec2010signed}, email networks~\cite{leskovec2007graph}, and voting networks~\cite{mucha2010communities}, are empowered by graphs and graph analysis~\cite{sharma2019brief}. For instance, Facebook leverages social network analysis to offer friend recommendations based on the connections between different users~\cite{jiang2016big}. While the benefits are indisputable, direct publication of graph data potentially results in individual privacy being exposed by different types of privacy attacks~\cite{hay2011privacy}. Hence, it is crucial to sanitize graph data before making it publicly available.

Differential privacy (DP)~\cite{dwork2014algorithmic} is an extensively studied statistical privacy model thanks to its rigorous mathematical privacy framework. DP can be applied to graph data in two common ways~\cite{hay2009accurate}: edge-level DP and node-level DP. Edge-level DP considers two graphs as neighbors if they differ by a single edge, while node-level DP considers two graphs as neighbors if they differ by the edges connected to a single node. Satisfying node-level DP can be challenging as varying one node could result in the removal of $N-1$ edges in the worst case, where $N$ denotes the number of nodes. 
With the development of differentially private deep learning~\cite{abadi2016deep}, most existing methods~\cite{yang2020secure, daigavane2021node,zhang2024dpar,sajadmanesh2023gap,xiang2023preserving} focus on generating synthetic graph data by privatizing deep graph generation models. The advanced composition mechanisms, such as moments accountant (MA)~\cite{abadi2016deep}, are employed to address excessive splitting of privacy budget during optimization, ensuring that the focus remains on mitigating high sensitivity in graph data. 
For example, Yang {\it et al}.~\cite{yang2020secure} propose two solutions, namely differentially private GAN (DPGGAN) and differentially private VAE (DPGVAE), which address high sensitivity by enhancing MA. However, they only achieve weak edge-level DP.
Recently, a number of methods~\cite{daigavane2021node,zhang2024dpar,sajadmanesh2023gap} focus on Graph Neural Networks (GNNs) under node-level RDP, which often mitigate high sensitivity by bounded-degree strategies.  
Despite their success, these methods achieve DP based on advanced composition mechanisms and tend to converge prematurely, particularly when the privacy budget is small. This issue results in decreased performance in terms of both privacy and utility. 

Over the past two decades, PageRank~\cite{page1998pagerank} has been widely used in graph mining and learning to evaluate node rankings, thanks to its simplicity, effectiveness, and ease of analysis. 
In this paper, we propose a novel privacy-preserving deep PageRank approach for graph synthesis, namely PrivDPR. A naive method for ensuring privacy is to clip the gradients for all weights and then add noise to them. However, this method often suffers from high sensitivity due to complex node relationships, leading to poor utility. To tackle this issue, inspired by the fact that graph synthesis only requires a specific weight instead of all weights, we explore the relationship between the number of layers and sensitivity. Instead of directly reducing high sensitivity, our core idea is to use weight normalization as a bridge to theoretically demonstrate that increasing the number of layers effectively addresses the challenges associated with high sensitivity and privacy budget splitting. Moreover, through formal privacy analysis, we provide evidence that the synthetic graphs generated by PrivDPR satisfy node-level DP requirements. Extensive experiments on four real graph datasets highlight PrivDPR's ability to effectively preserve essential structural properties of the original graphs.

To summarize, this paper makes three main contributions:
\begin{itemize}
[leftmargin=5mm]
\item We present PrivDPR, a novel method for privately synthesizing graphs through the design of a deep PageRank. This method can preserve high data utility while ensuring $(\epsilon,\delta)$-node-level DP.
\item Instead of directly reducing high sensitivity, we utilize weight normalization as a bridge to explore the relationship between the number of layers and sensitivity. We reveal that increasing the number of layers can effectively address issues arising from high sensitivity and privacy budget splitting.
\item Extensive experiments on real graph datasets demonstrate that our solution outperforms existing state-of-the-art methods significantly across eight distinct graph utility metrics and two classical downstream tasks.
\end{itemize}

The remaining sections of this paper are organized as follows. In Section~\ref{Preliminary}, we provide the preliminaries of our proposed solution.
Section~\ref{ProDef} presents the problem definition and introduces existing solutions.
In Section~\ref{PrivPRM_Framwork}, we discuss the sanitization solution. The privacy and time complexity analysis are presented in Section~\ref{sec:PrivPRM_PrivCompAnalysis}.
The comprehensive experimental results are presented in Section~\ref{Experiment}. We review the related work in Section~\ref{Related_work}. Finally, we conclude this paper in Section~\ref{ConcluSection}. 
\section{Preliminaries}\label{Preliminary}
In this section, we will provide a concise review of the concepts of DP and PageRank. 
The mathematical notations used throughout this paper are summarized in Table~\ref{PrivPRM_SymbolTable}. 
% in \textbf{Appendix~\ref{appendix:notations}}.

\begin{table}[htb]\small
  \centering
  \begin{threeparttable}
      \caption{Common Symbols and Definitions}\label{PrivPRM_SymbolTable}
      \begin{tabular}{l|l}
        \hline   %  or \cline{col1-col2}
        Symbol & Description \\
        \hline   %  or \cline{col1-col2}
        {$\epsilon, \delta$} & {Differential privacy parameters} \\
        {$G, \widetilde{G}$} & {Original and synthetic graphs}  \\
        % {$\mathbf{A}, \widetilde{\mathbf{A}}$} & {Original and synthetic adjacency matrices} \\
        % {$\mathbf{A}$} & {The adjacency matrix}  \\
        {$G, G^\prime$} & {Any two neighboring graph datasets} \\
        {$d_i^{in}, d_i^{out}$} & {In-degree and out-degree of node $i$} \\
        {$\mathcal{A}$} & {A randomized algorithm}                       \\
        {$V$, $E$} & {Set of nodes and edges of $G$} \\
        % {$E$} & {Set of edges of $G$} \\
        {$N$}  & {Number of nodes in $G$} \\
        {$\mathbf{x},\mathbf{y}$} & {Lowercase letters denote vectors}  \\
        {$\mathbf{X},\mathbf{Y}$} & {Bold capital letters denote matrices}  \\
        % {\textcolor{red}{$\|\mathbf{X}\|_F$}} & {\textcolor{red}{Frobenius norm}} \\
        {$\|\mathbf{x}\|_2, \|\mathbf{X}\|_2$} & {$\ell_2$-norm and Spectral norm} \\
        {$r$} & {Dimension of low-dimensional vectors} \\
        {$\gamma$} & {Damping factor of PageRank model} \\
        \hline   %  or \cline{col1-col2}
      \end{tabular}
  \end{threeparttable}
\end{table}

\subsection{Differential Privacy}
DP is the prevailing concept of privacy for algorithms on statistical databases. Informally, DP limits the change in output distribution of a mechanism when there is a slight change in its input. In graph data, the concept of neighboring databases is established using two graph datasets, denoted as $G$ and $G^\prime$. These datasets are considered neighbors if their dissimilarity is limited to at most one edge or node. 

\begin{definition}[Edge (Node)-Level DP~\cite{hay2009accurate}]\label{DP_def}
A graph analysis mechanism $\mathcal{A}$ achieves $(\epsilon, \delta)$-edge (node)-level DP, if for any pair of input graphs $G$ and $G^\prime$ that are neighbors (differ by at most one edge or node), and for all possible $O\subseteq Range(\mathcal{A})$, we have 
$\mathbb{P}[\mathcal{A}(G) \in O] \leq \exp(\epsilon) \cdot \mathbb{P}[\mathcal{A}(G^\prime) \in O] + \delta$.
  % \begin{equation*}
  %   \begin{split}
  %     \mathbb{P}[\mathcal{A}(G) \in O] \leq \exp(\epsilon) \cdot \mathbb{P}[\mathcal{A}(G^\prime) \in O] + \delta.
  %   \end{split}
  % \end{equation*}
\end{definition}

The concept of the neighboring dataset $G$, $G^\prime$ is categorized into two types. Specifically, if $G^\prime$ can be derived by replacing a single data instance in $G$, it is termed bounded DP~\cite{dwork2006calibrating}. If $G^\prime$ can be obtained by adding or removing a data sample from $G$, it is termed unbounded DP~\cite{dwork2006differential}.
The parameter $\epsilon$ is referred to as the privacy budget, which is utilized to tune the trade-off between privacy and utility in the algorithm. A smaller value of $\epsilon$ indicates a higher level of privacy protection. The parameter $\delta$ is informally considered as a failure probability and is typically selected to be very small.

Suppose that a function $f$ maps a graph $G$ to a $r$-dimensional output in $\mathbb{R}^r$. To create a differentially private mechanism from $f$, it is common practice to inject random noise into the output of $f$. The magnitude of this noise is determined by the sensitivity of $f$, defined as follows.
\begin{definition}[Sensitivity~\cite{dwork2006calibrating}]\label{SensDef}
Given a function $f: G\rightarrow\mathbb{R}^r$, for any neighboring datasets $G$ and $G^\prime$, the $\ell_2$-sensitivity of $f$ is defined as 
$\mathcal{S}_{f} = \max_{G, G^\prime}\|f(G) - f(G^\prime)\|_2$.
% \begin{equation*}
%   \Delta_{f} = \max_{G, G^\prime}\|f(G) - f(G^\prime)\|_2.
% \end{equation*}
\end{definition}

\textbf{Gaussian mechanism.} By utilizing the $\ell_2$-sensitivity definition, we can formalize the Gaussian mechanism applied to $f$ as follows:
\begin{theorem}[Gaussian Mechanism~\cite{dwork2006calibrating}]\label{GauMech}
For any function $f: G\rightarrow\mathbb{R}^r$, the Gaussian mechanism is defined as
$\mathcal{A}(G) = f(G) + \mathcal{N}\left(\mathcal{S}_{f}^2\sigma^2\mathbf{I}\right)$,
% \begin{equation*}
% \mathcal{A}(G) = f(G) + \mathcal{N}\left(\mathcal{S}_{f}^2\sigma^2\mathbf{I}\right),
% \end{equation*}
where $\mathcal{N}\left(\mathcal{S}_{f}^2\sigma^2\mathbf{I}\right)$ represents a zero-mean Gaussian distribution with $\sigma = \frac{\sqrt{2\log(1.25/\delta)}}{\epsilon}$.
\end{theorem}

The concept of sensitivity implies that protecting privacy at the level of individual nodes is more challenging compared to the edge level. This is primarily because modifying a node typically has a considerably greater impact (higher sensitivity) than changing an edge. As a result, a significant amount of noise must be added to ensure privacy for individual nodes.
% \begin{remark}
% The concept of sensitivity implies that protecting privacy at the level of individual nodes is more challenging compared to the edge level. This is primarily because modifying a node typically has a considerably greater impact (higher sensitivity) than changing an edge. As a result, a significant amount of noise must be added to ensure privacy for individual nodes.
% \end{remark}

\textbf{Important properties.} Furthermore, DP possesses two important properties that play a significant role in the implementation of PrivDPR.

\begin{theorem}[Sequential Composition~\cite{dwork2014algorithmic}]\label{DP_compos}
If $\mathcal{A}_1$ ensures $(\epsilon_1,\delta_1)$-DP, and $\mathcal{A}_2$ ensures $(\epsilon_2,\delta_2)$-DP, then the composition $(\mathcal{A}_1\circ\mathcal{A}_2)$ guarantees $(\epsilon_1+\epsilon_2,\delta_1+\delta_2)$-DP.
\end{theorem}

\begin{theorem}[Post-Processing~\cite{dwork2014algorithmic}]\label{PostProc_Theo}
If $\mathcal{A}$ is an algorithm that achieves $(\epsilon,\delta)$-DP, then the sequential composition $\mathcal{B}(\mathcal{A}(\cdot))$ with any other algorithm $\mathcal{B}$ that does not have direct or indirect access to the private database also satisfies $(\epsilon,\delta)$-DP.
\end{theorem}

\subsection{PageRank}\label{Prelim_PR}
The Internet and social networks can be seen as vast graph structures. 
PageRank~\cite{page1998pagerank} is a well-known algorithm used for analyzing the links in a graph, making it a representative method for graph link analysis. It operates as an unsupervised learning approach on graph data. The core concept of PageRank involves establishing a random walk model on a directed graph, which can be viewed as a first-order Markov chain. This model describes the behavior of a walker randomly visiting each node along the directed edges of the graph. By meeting certain conditions, the probability of visiting each node during an infinitely long random walk converges to a stationary distribution. At this point, the stationary probability assigned to each node represents its PageRank value, indicating its significance. PageRank is defined recursively, and its calculation is typically performed using an iterative algorithm. The formal definition of PageRank is as follows:
\begin{lemma}
Consider a graph $G$. The PageRank score of a node $j$, denoted as $PR_j$, represents the probability of reaching node $j$ through random walks. The value of $PR_j$ can be calculated by summing up the ranking scores of its direct predecessors $i$, weighted by the reciprocal of their out-degree $d_i^{out}$. Mathematically, we have:
    \begin{equation}\label{PR_lossFunc}
         PR_j = \gamma\left(\sum_{i \in P_j} \frac{PR_i}{d_i^{out}}\right)+\frac{1-\gamma}{N}, \quad j=1,2, \cdots, N
    \end{equation}
where $\gamma$ is a damping factor commonly set to 0.85, $P_j$ represents the set of direct predecessors of node $j$, and the second term ensures that $PR_j > 0$ for $j=1,2, \cdots, N$, with $\sum_{j=1}^N PR_j = 1$.
\end{lemma}

% PageRank, originally designed for directed graphs, can be effectively applied to undirected graphs by duplicating the training process. This adaptation maintains the algorithm's generality without any compromise. 
%\textcolor{red}{allowing for the possibility of adding or removing edges or nodes.}
\section{Problem Definition and Existing Solutions}\label{ProDef}
\subsection{Problem Definition}
In this paper, we consider a directed and unweighted graph $G = (V, E)$, where $V$ represents the set of nodes and $E$ represents the set of edges. An undirected graph is a special case of the graph that we have defined and is included in our problem definition. Our primary objective is to address the following problem: \emph{given a graph $G$, how can we generate a synthetic graph $\widetilde{G}$ that possesses similar graph properties as the original graph $G$, while ensuring node-level DP.} 
 
\begin{definition}[Graph Synthesis under Bounded DP\footnote{In this paper, we achieve node-level privacy protection for graph synthesis under bounded DP. This implies that the synthesized graph can have the same number of nodes as the original graph.}]
A graph synthesis model $\mathscr{L}$ satisfies $(\epsilon, \delta)$-node-level DP if two neighboring graphs $G$ and $G^\prime$, which differ in only a node and its corresponding edges, satisfy the following condition 
% for all sets of possible synthetic graphs $\widetilde{G}_s$:
for all possible $\widetilde{G}_s\subseteq{Range(\mathscr{L})}$:
\begin{equation*}
  \mathbb{P}(\mathscr{L}(G) \in \widetilde{G}_s) \leq \exp(\epsilon) \cdot \mathbb{P}\left(\mathscr{L}\left(G^\prime\right) \in \widetilde{G}_s\right) + \delta,
\end{equation*}
where $\widetilde{G}_s$ denotes the set comprising all possible $\widetilde{G}$.
\end{definition}

The generated synthetic graph can be utilized for various downstream graph analysis tasks without compromising privacy, thanks to the post-processing property of DP.
\begin{theorem}
Let $\mathscr{L}$ be an $(\epsilon,\delta)$-node-level private graph synthesis model, and $f$ is an arbitrary graph query whose input is a simple graph. Then, $f \circ \mathscr{L}$ satisfies $(\epsilon,\delta)$-node-level DP.
\end{theorem}

In this scenario, even though an attacker has knowledge of the differentially private protocol, encompassing the methods of data encoding and perturbation, they are unable to deduce the original information accurately.

\subsection{Existing Solutions}
Existing solutions for graph generation include two tracks: differentially private shallow graph models and differentially private deep graph learning models. As our focus is on the latter one, we will defer the introduction of private shallow models to Section~\ref{Related_work}. In traditional differentially private deep learning~\cite{abadi2016deep,fu2024dpsur}, the advanced
composition mechanisms (i.e., MA) are employed to address excessive splitting of privacy budget during optimization. DP optimizers for \emph{non-graph data} typically update on the summed gradient with Gaussian noise:
\begin{equation}\label{dpsgd_eq1}
  \mathbf{\tilde{g}} \leftarrow \frac{1}{B}\left(\sum_{i=1}^{B} Clip(\mathbf{g}\left(x_i\right))+\mathcal{N}\left(C^2 \sigma^2 \mathbf{I}\right)\right),
\end{equation}
where $\mathbf{g}\left(x_i\right)$ denotes the gradient for each example $x_i$, $Clip(\cdot)$ is the clipping function defined as $Clip(\mathbf{g}\left(x_i\right))=\mathbf{g}\left(x_i\right)/\max\left(1,\frac{\|\mathbf{g}\left(x_i\right)\|_2}{C}\right)$,
$C$ is the clipping threshold, and $B$ is the batch size. 
% Unlike traditional differential privacy data analysis~\cite{ye2019privkv,ye2021beyond,ye2021privkvm,ye2023stateful,Zhang2023trajectory}, 
In graph learning, individual examples no longer compute their gradients independently because changing a single node or edge in the graph may affect all gradient values. To address this issue, three types of solutions have been proposed:
\begin{itemize}
[leftmargin=5mm]
\item \emph{Private GAN model.} 
DPGGAN~\cite{yang2020secure} is a differentially private GAN for graph synthesis. It improves the sensitivity by enhancing the MA. DPGGAN proves that the noised clipped gradient $\mathbf{\tilde{g}}$ applied as above guarantees that the learned graph generation model to be edge-level DP, with a different condition from that in Theorem~\ref{GauMech} due to the nature of graph generation.
\item \emph{Private VAE model.} 
DPGVAE~\cite{yang2020secure} is a differentially private VAE designed for graph synthesis. It achieves the same level of privacy as DPGGAN under the same conditions.
\item \emph{Private GNN model.} 
The features of one node can influence the gradients of other nodes in the network. The sensitivity for $\sum_{i=1}^{B} Clip(\mathbf{g}\left(x_i\right))$ may reach $BC$ under node-level DP. Several solutions~\cite{sajadmanesh2023gap,zhang2024dpar,daigavane2021node,olatunji2021releasing,xiang2023preserving} have been proposed. The classic and advanced approach is GAP~\cite{sajadmanesh2023gap}, which uses aggregation perturbation to achieve RDP and introduces a new GNN architecture tailored for private learning over graphs, resulting in improved trade-offs between privacy and accuracy.
\end{itemize}

\textbf{Limitations.} 
Despite the usefulness of these approaches, two limitations in generating synthetic graphs have yet to be solved:
1) DPGGAN and DPGVAE only achieve weak edge-level DP, and
2) MA and RDP typically require a sufficient privacy budget to estimate the privacy guarantee. Thus, these private models mentioned above tend to converge
prematurely with a small privacy budget. 
We explain the second issue in detail using Algorithm~\ref{DPSGD_Algo}.

In the deep learning with DP framework, described in the algorithm from Line~\ref{code:Accou_begin} to Line~\ref{code:Accou_end}, there is a potential issue of premature termination when working with a limited overall privacy budget, such as $\epsilon \leq 0.5$. This occurs because the privacy loss metric, $\delta_{Accountant}$, converges rapidly towards the desired privacy level $\delta$. As a result, the algorithm may stop prematurely before achieving the desired level of privacy. This premature termination poses a challenge for existing deep learning models that utilize this framework. To mitigate this issue and improve the utility of the trained models, it is often necessary to set a relatively large value for the privacy budget parameter, $\epsilon$. However, this approach introduces a trade-off between privacy and utility since a larger privacy budget carries a higher risk to data privacy.

\begin{algorithm}
\caption{Deep Learning with DP}\label{DPSGD_Algo}
\For{epochs = $1,\cdots,T$}
{
    Take a batch sample set with sampling probability $p$\;
    Apply clipping to per-sample gradients\;
    Add Gaussian noise to the sum of clipped gradients\;\label{code:Add_noise}
    Update weights by any optimizer on private gradients with learning rate $\eta$\;
    Compute $\delta_{Accountant}$ given the target $\epsilon$\;
    \If{$\delta_{Accountant}(\epsilon,\sigma,p,T)\geq\delta$}
    { \label{code:Accou_begin}
         Break\;
    } \label{code:Accou_end}

}
\end{algorithm}
\section{Our Proposal: PrivDPR}\label{PrivPRM_Framwork}
To tackle the limitations outlined in the previous section, we propose a node-level differentially private deep PageRank for graph synthesis, inspired by the simplicity, effectiveness, and ease of analysis of PageRank in Section~\ref{Prelim_PR}. First, we provide an overview of the approach. Next, we describe how we construct the deep PageRank and achieve gradient perturbation. Finally, we present the complete training algorithm. 

\subsection{Overview}
The workflow of PrivDPR is shown in Figure~\ref{PrivPRM_FWFig}, which consists three stages: deep PageRank, gradient perturbation, and graph reconstruction. 
\begin{itemize}
[leftmargin=5mm]
    \item \textbf{Deep PageRank.} We design a deep PageRank that serves as the foundation for analyzing the correlation between the number of layers, high sensitivity, and privacy budget splitting. (see Section~\ref{sec:Deep_PRM})
    \item \textbf{Gradient Perturbation.} We achieve private deep PageRank by gradient perturbation. Instead of directly reducing high sensitivity, we first reveal theoretically that we can preset a desired small sensitivity and achieve it by slightly increasing the number of layers. We then show that the theorem can be extended to resist privacy budget splitting. (see Section~\ref{sec:gra_perturb})
    \item \textbf{Graph Reconstruction.} We reconstruct the graph by examining the co-occurrence counts of nodes using the acquired representations during optimization. This process entails creating a transition count matrix, and deriving an edge probability matrix to produce a binary adjacency matrix that represents the synthesized graph. (see \textbf{Appendix~\ref{appendix:gra_syn}})
\end{itemize}

% Furthermore, model optimization is presented in Section~\ref{sec:model_opt}.
\begin{figure*}[htb]
  \centerline{\includegraphics[width=6.5in]{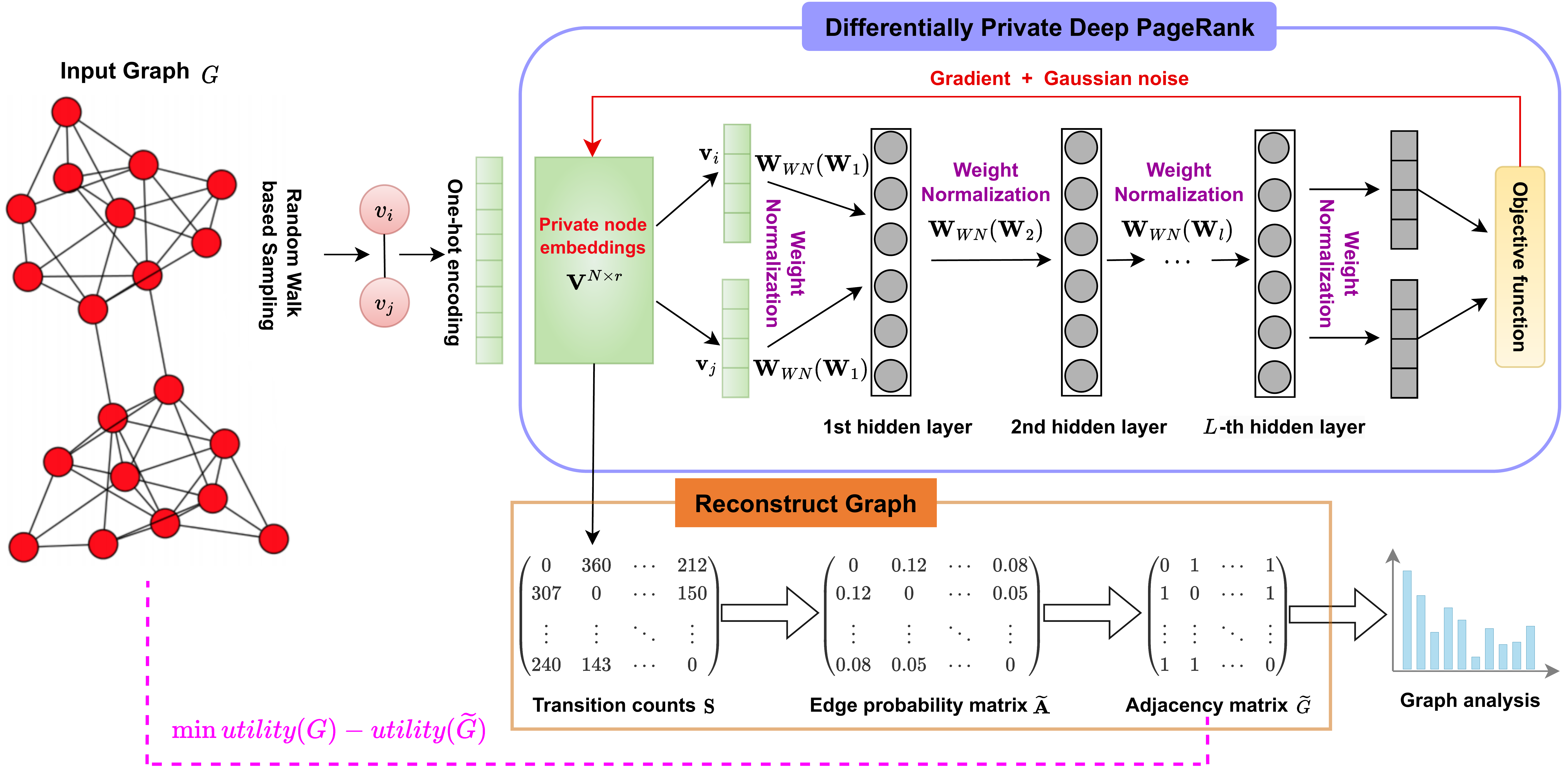}}
  \caption{Framework of our proposed PrivDPR}\label{PrivPRM_FWFig}
\end{figure*}

\subsection{Deep PageRank}\label{sec:Deep_PRM}
We fuse multiple layers into the PageRank, resulting in a modified form of Eq.~(\ref{PR_lossFunc}) given by:
\begin{equation}\label{DPRObjFunc}
    \begin{split}
        \min_{\boldsymbol{\Theta}} \mathscr{L}
        =
        \sum_{j \in V}\left(\gamma\left(\sum_{i \in P_j} \frac{f(v_i; \boldsymbol{\Theta})}{d_i^{out}}\right)+\frac{1-\gamma}{N} - f(v_j; \boldsymbol{\Theta})\right)^2,
    \end{split}
\end{equation}
where $f(v; \boldsymbol{\Theta})$ is a fully connected neural network in the following form:
\begin{equation}\label{NeuralNet_model}
\begin{split}
     f(v; \boldsymbol{\Theta})
    =\phi_{L+1}\left(\mathbf{W}_{L+1} \phi_L\left(\mathbf{W}_L\left(\phi_{L-1}\left(\mathbf{W}_{L-1}\left(\ldots \phi_1\left(\mathbf{V}\mathbf{W}_1\right) \ldots\right)\right)\right)\right)\right).
\end{split}
\end{equation}

Figure~\ref{PrivPRM_FWFig} shows the detailed architecture of deep PageRank (i.e., Eq.~(\ref{DPRObjFunc})), in which graph data is fed into the neural network through a virtual one-hot encoding of nodes as $\mathbf{V}$. The set of learning parameters is denoted as $\Theta = \{\mathbf{V},\mathbf{W}_1,\cdots,\mathbf{W}_l,\mathbf{W}_{L+1}\}$, with $\mathbf{V}$ in dimensions of $N\times r$, $\mathbf{W}_1$ in dimensions of $r\times d$, $\mathbf{W}_l$ in dimensions of $d\times d$, and $\mathbf{W}_{L+1}$ in dimensions of $d\times 1$. The activation function employed in each layer is denoted as $\phi$. For simplicity, the bias terms of each layer are omitted.

However, in the backpropagation of deep PageRank, applying stochastic gradient descent (SGD) to update $\mathscr{L}$ becomes infeasible,
since the squared loss term in Eq.~(\ref{DPRObjFunc}) involves a summation over all nodes $i$ pointing to node $j$, denoted as $\sum_{i \in P_j}$. This means that each squared loss term aggregates information from multiple links pointing to the same node $j$, which contradicts the standard SGD assumption (i.e., $\sum_{(i,j)\in E} L_{ij}$) and is thus non-decomposable. 
To address this challenge, we alternatively establish an upper bound for $\mathscr{L}$. 

\begin{lemma}\label{lemma_uppbound_objfunc}
By applying the Cauchy-Schwarz inequality to Eq.~(\ref{DPRObjFunc}), an upper bound of the objective function is:
\begin{equation*}
  \begin{split}
      \min_{\boldsymbol{\Theta}} \mathscr{L}
      &\leq \sum_{(i,j) \in E}d_j^{in}\gamma^2\left(\frac{f(v_i; \boldsymbol{\Theta})}{d_i^{out}}- \frac{f(v_j; \boldsymbol{\Theta})}{d_j^{in}\gamma}\right)^2 \\
      &+\sum_{(i,j) \in E}\left(\frac{f(v_i; \boldsymbol{\Theta})}{d_i^{out}}- \frac{f(v_j; \boldsymbol{\Theta})}{d_j^{in}\gamma}\right)\frac{2\gamma(1-\gamma)}{N} + \sum_{(i,j) \in E}\frac{(1-\gamma)^2}{d_j^{in}N^2}.
  \end{split}
\end{equation*}
\begin{proof}
Please refer to \textbf{Appendix~\ref{appendix:lemma_proof}}.
\end{proof}
\end{lemma}

According to Lemma~\ref{lemma_uppbound_objfunc}, the objective function for each edge $(i,j)$ is given by:
\begin{equation}\label{PRM_lossFunc}
\begin{split}
     &\mathcal{L}(v_i, v_j; \boldsymbol{\Theta})
    = d_j^{in}\gamma^2\left(\frac{f(v_i; \boldsymbol{\Theta})}{d_i^{out}}- \frac{f(v_j; \boldsymbol{\Theta})}{d_j^{in}\gamma}\right)^2 \\
     &+ \left(\frac{f(v_i; \boldsymbol{\Theta})}{d_i^{out}}- \frac{f(v_j; \boldsymbol{\Theta})}{d_j^{in}\gamma}\right)\frac{2\gamma(1-\gamma)}{N}
     + \frac{(1-\gamma)^2}{d_j^{in}N^2}.
\end{split}
\end{equation}

The proof of this upper bound on the approximation ratio is currently a subject for future research. However, it is important to note that the effectiveness of this upper bound has been demonstrated through experiments, which will be discussed and presented later on.

\subsection{Gradient Perturbation}\label{sec:gra_perturb}
\subsubsection{\textbf{How to Resist High Sensitivity?}}\label{resist_sens}
To yield a private embedding matrix $\mathbf{V}$, a naive method is to first clip $\frac{\partial{\mathcal{L}(v_i, v_j; \boldsymbol{\Theta})}}{\partial{\mathbf{V}}}$ and then inject noise into this gradient. This results in the following expression: 
\begin{equation}\label{dpsgd_eq2}
  \widetilde{\nabla}_{\mathbf{V}}\mathcal{L}
  \leftarrow \frac{1}{B}\left(\sum_{(v_i,v_j) \in E_B}Clip\left(\frac{\partial{\mathcal{L}(v_i, v_j; \boldsymbol{\Theta})}}{\partial{\mathbf{V}}}\right)+\mathcal{N}\left(\mathcal{S}_\nabla^2 \sigma^2 \mathbf{I}\right)\right),
\end{equation}
where the sensitivity of $\sum_{(v_i,v_j) \in E_B}Clip\left(\frac{\partial{\mathcal{L}(v_i, v_j; \boldsymbol{\Theta})}}{\partial{\mathbf{V}}}\right)$, denoted as $\mathcal{S}_\nabla$, can reach up to $BC$, as modifying one node could potentially impact all gradients in Eq.~(\ref{dpsgd_eq2}).  

As shown in Figure~\ref{PrivPRM_FWFig}, we design an optimizable matrix $\mathbf{V}$ as the input to the neural network $f(v; \boldsymbol{\Theta})$. Since our goal is to privatize $\mathbf{V}$ for graph synthesis, we only need to add noise to the gradient of $\mathbf{V}$, rather than to the gradients of all the weights. Inspired by this, we explore the use of weight normalization in $f(v; \boldsymbol{\Theta})$, which enables us to naturally bound the gradient of $\mathbf{V}$ and further reveal the relationships among the model's parameters.
Using weight normalization (see Figure~\ref{PrivPRM_FWFig}), not $Clip(\cdot)$, as a bridge to bound $\frac{\partial{\mathcal{L}(v_i, v_j; \boldsymbol{\Theta})}}{\partial{\mathbf{V}}}$, we can rewrite Eq.~(\ref{dpsgd_eq2}) as follows:
\begin{equation}\label{dpsgd_eq3}
\widetilde{\nabla}_{\mathbf{V}}\mathcal{L} \leftarrow \frac{1}{B}\left(\sum_{(v_i,v_j) \in E_B} \frac{\partial{\mathcal{L}(v_i, v_j; \boldsymbol{\Theta})}}{\partial{\mathbf{V}}}+\mathcal{N}\left(\mathcal{S}_\nabla^2 \sigma^2 \mathbf{I}\right)\right).
\end{equation}

Theorem~\ref{PrivPRM_MainTheo} reveals that we can preset a desired small $\mathcal{S}_\nabla$ for $\sum_{(v_i,v_j) \in E_B}\frac{\partial{\mathcal{L}(v_i, v_j; \boldsymbol{\Theta})}}{\partial{\mathbf{V}}}$ and achieve it by slightly increasing the number of layers.

\begin{theorem}\label{PrivPRM_MainTheo}
Given $\Theta = \{\mathbf{V},\mathbf{W}_l\}_{l=1}^{L+1}$ with weight normalization using $\mathbf{W}_{WN}(\mathbf{W}) = \frac{\mathbf{W}}{\|s\mathbf{W}\|_2}$ where $s>1$, we have 
$\left\|\frac{\partial{\mathcal{L}}(v_i, v_j; \boldsymbol{\Theta})}{\partial{\mathbf{V}}}\right\|_2 \leq M\left(\frac{1}{s}\right)^{L+1}$, 
in which $M = \left(2(N-1)\gamma^2 + 2\gamma + \frac{2\gamma(1-\gamma)}{N}\right)\left(1 + \frac{1}{\gamma}\right)$. By presetting a sensitivity $\mathcal{S}_\nabla$, we can determine the maximum number of layers $L$ using $\log_\frac{1}{s}^\frac{\mathcal{S}_\nabla}{BM} - 1 \leq L$ under node-level DP.
\end{theorem}

To prove Theorem~\ref{PrivPRM_MainTheo}, we will need the following lemmas.

\begin{lemma}\label{lemma_NNObj_On_V}
For activation functions with bounded derivatives, such as Sigmoid, the upper bound of $\frac{\partial f(v; \boldsymbol{\Theta})}{\partial{\mathbf{V}}}$ is given by:
\begin{equation}\label{NeuralNetGra_On_V}
    \begin{split}
         \left\|\frac{\partial f(v; \boldsymbol{\Theta})}{\partial\mathbf{V}}\right\|_2
         \leq \prod_{l=1}^{L+1}  \left\|\mathbf{W}_l\right\|_2,
    \end{split}
\end{equation}
where $\frac{\partial f(v; \boldsymbol{\Theta})}{\partial{\mathbf{V}}}=\frac{\partial f(v; \boldsymbol{\Theta})}{\partial{\mathbf{v}}}$ holds because the input layer uses a one-hot-encoded vector and $\mathbf{v}$ denotes a vector from $\mathbf{V}$.
\begin{proof}
Please refer to \textbf{Appendix~\ref{appendix:lemma2_proof}}.
\end{proof}
\end{lemma}

Next, we give the upper bound of $\frac{\partial{\mathcal{L}(v_i, v_j; \boldsymbol{\Theta})}}{\partial{\mathbf{V}}}$.

\begin{lemma}\label{lemma_BatchObjFuncGra_On_V}
The upper bound of $\frac{\partial{\mathcal{L}(v_i, v_j; \boldsymbol{\Theta})}}{\partial{\mathbf{V}}}$ is
\begin{equation}\label{BatchObjFuncGra_On_V}
    \begin{split}
        \left\|\frac{\partial{\mathcal{L}(v_i, v_j; \boldsymbol{\Theta})}}{\partial{\mathbf{V}}}\right\|_2
      \leq M\prod_{l=1}^{L+1}\left\|\mathbf{W}_l\right\|_2,
    \end{split}
\end{equation}
where $M = \left(2(N-1)\gamma^2 + 2\gamma + \frac{2\gamma(1-\gamma)}{N}\right)\left(1 + \frac{1}{\gamma}\right)$.
\begin{proof} 
Please refer to \textbf{Appendix~\ref{appendix:lemma3_proof}}.
\end{proof}
\end{lemma}

We now prove the Theorem~\ref{PrivPRM_MainTheo}.
\begin{proof}[Proof of Theorem~\ref{PrivPRM_MainTheo}]
In Lemma~\ref{lemma_BatchObjFuncGra_On_V}, by normalizing $\mathbf{W}$ using $\mathbf{W}_{WN}(\mathbf{W}) = \mathbf{W} / \|s\mathbf{W}\|_2$ where $s>1$, we have 
\begin{equation}\label{FNGra_On_V}
    \begin{split}
        \left\|\frac{\partial{\mathcal{L}(v_i, v_j; \boldsymbol{\Theta})}}{\partial{\mathbf{V}}}\right\|_2
        \leq M\prod_{l=1}^{L+1}\left\|\mathbf{W}_{WN}(\mathbf{W}_l)\right\|_2 
        \leq M\left(\frac{1}{s}\right)^{L+1}.
    \end{split}
\end{equation}

Consider the worst case, where every gradient is affected in Eq.~(\ref{dpsgd_eq3}) under node-level DP. Then, according to Eq.~(\ref{FNGra_On_V}), we have
\begin{equation}\label{Bound_by_C_g}
  \begin{split}
    BM\left(\frac{1}{s}\right)^{L+1} \leq \mathcal{S}_\nabla,
  \end{split}
\end{equation}
where $\mathcal{S}_\nabla$ is a desired sensitivity that can be preset. Therefore, we have
\begin{equation}\label{FNTheorem_result}
    \begin{split}
        \log_\frac{1}{s}^\frac{\mathcal{S}_\nabla}{BM} - 1
    \leq L,
    \end{split}
\end{equation}
which shows that slightly increasing the number of layers can effectively overcome the high sensitivity.
\end{proof}
The following is a real example showing the calculation to achieve a preset sensitivity.
\begin{example}
Citeseer~\cite{sen2008collective} is a popular citation network with 3,327 nodes and 4,732 edges, which is used as the input data.
Recall $f(v; \boldsymbol{\Theta})$ with $\mathbf{V}\in\mathbb{R}^{N \times r}$, $\mathbf{W}^l\in\mathbb{R}^{r \times d}$, and $\mathbf{W}^{L+1}\in\mathbb{R}^{d \times 1}$.
Using $N=3,327, \gamma=0.85$, we have $M\approx10,464$. Then, with $r=d=128, \mathcal{S}_\nabla=5, B=128, s=5$, we determine that $7 \leq L$ using Eq.~(\ref{FNTheorem_result}).
\end{example}

\subsubsection{\textbf{How to Resist Privacy Budget Splitting?}}\label{resist_privbug} 
Given the total privacy parameters $\epsilon$ and $\sigma$, and the total number of iterations $T$ for model optimization,
the sequential composition property of DP in Theorem~\ref{DP_compos} requires dividing both $\epsilon$ and $\delta$. 
Instead of using advanced composition mechanisms, we evenly divide privacy parameters $\epsilon$ and $\sigma$. In particular, we calculate $\sigma$ using $\sigma = \frac{\sqrt{2\log(1.25/(\delta/T))}}{\epsilon/T}$.
Note that $\delta$ is incorporated within a logarithm function, making its impact negligible. Therefore, we only focus on how to resist privacy budget splitting. 
In fact, Theorem~\ref{PrivPRM_MainTheo} provides a perspective to tackle this issue by reducing $\mathcal{S}_\nabla$ in Eq.~(\ref{Bound_by_C_g}), which is defined as $BM\left(\frac{1}{s}\right)^{L+1} \leq \frac{\mathcal{S}_\nabla}{T}$.
This leads to the following inequality: 
\begin{equation}\label{ResPrivBug}
    \begin{split}
        \log_\frac{1}{s}^\frac{\mathcal{S}_\nabla}{BMT} - 1
    \leq L.
    \end{split}
\end{equation}

Note that $\log(\cdot)$ has a slow growth rate, which means that the additional cost of resistance against privacy budget splitting is not significant.

\begin{algorithm}
\caption{PrivDPR Algorithm}\label{PrivPRM_Alg}
\KwIn{
graph $G$,
privacy parameters $\epsilon$ and $\delta$,
delay factor $\gamma$,
preset sensitivity $\mathcal{S}_\nabla$,
number of training epochs $n_{epochs}$,
batch size $\flat$,
random walk number $R_{wn}$,
random walk length $R_{wl}$,
dimensions $r$ and $d$,
learning rate $\eta$.
}
\KwOut{Synthetic graph $\widetilde{G}$.}
Initialize the learning parameters set\;
\For{$i=1$ to $n_{epochs}$}
{
    \For{$j=1$ to $\lfloor{N / \flat}\rfloor$}
    {
        \tcp{Generate batch samples}
        Create a list of indices, $node\_list$, starting from $j \cdot \flat$ and ending at $(j + 1) \cdot \flat - 1$\;
        \For{each $node\_id$ in node\_list}
        {
        \label{AlgCode_RW_start}
        Generate a batch of node pairs $E_B$ through random walks with walk number $R_{wn}$ and walk length $R_{wl}$\; 
        }
        \label{AlgCode_RW_end}
        \tcp{See Section~\ref{resist_sens} for details}
        Apply weight normalization to each weight $\mathbf{W}$\;
        \label{AlgCode_SpectNorm}
        Update $\mathbf{W}$ by Adam optimizer with learning rate $\eta$\; \label{AlgCode_update_W}
        \tcp{See Sections~\ref{resist_sens} and \ref{resist_privbug} for details}
        Add Gaussian noise to the sum of the gradients for $\mathbf{V}$\; \label{AlgCode_V_AddNoise}
        Update $\mathbf{V}$ by Adam optimizer with learning rate $\eta$\; \label{AlgCode_update_V}
        \tcp{Generate score matrix}
        Sample graphs from $\mathbf{V}\mathbf{V}^\top$ to generate score matrix $\mathbf{S}$\;
        \label{AlgCode_SampleGraphs}
    }
}
Convert score matrix $\mathbf{S}$ to edge-independent model $\widetilde{\mathbf{A}}$:
$\mathbf{S}^{\dagger} \leftarrow \max\{\mathbf{S}, \mathbf{S}^\top\}, \widetilde{\mathbf{A}} \leftarrow \mathbf{S}^{\dagger}/sum(\mathbf{S}^{\dagger}), \widetilde{G} \leftarrow \widetilde{\mathbf{A}}$\;
\label{AlgCode_GraphSyn}
\Return $\widetilde{G}$\;
\end{algorithm}

\subsection{Model Optimization}\label{sec:model_opt}
The pseudo-code of PrivDPR are presented in Algorithm~\ref{PrivPRM_Alg}. We first generate a batch set $E_B$ by random walk.
Next, we input these samples into the node embedding matrix $\mathbf{V}$ using a one-hot encoded vector with a length of $N$. The resulting low-dimensional vectors are then fed into a neural network.
To constrain the gradient with respect to $\mathbf{V}$, we normalize each weight $\mathbf{W}$ using weight normalization, and then update $\mathbf{W}$.
Subsequently, we introduce Gaussian noise to the sum of the gradients for $\mathbf{V}$, and then update $\mathbf{V}$. After each parameter update, we count transitions in score matrix $\mathbf{S}$. After finishing $n_{epochs}$ training, we transform $\mathbf{S}$ into edge probability matrix $\widetilde{\mathbf{A}}$.

\section{Privacy and Complexity Analysis}\label{sec:PrivPRM_PrivCompAnalysis}
% \textbf{Privacy Analysis.}
\subsection{Privacy Analysis}
In this section, we provide a privacy analysis for PrivDPR.
\begin{theorem}\label{GraPerAlg_Proof}
The synthetic graphs generated by PrivDPR satisfies $(\epsilon,\delta)$-node-level DP.
\begin{proof}
In Algorithm~\ref{PrivPRM_Alg}, for each weight parameter that needs optimization, the total number of iterations $T=n_{epochs}\lfloor{N / \flat}\rfloor$ is fixed a priori, and the desired privacy cost, say $\epsilon$, is split across the iterations: $\epsilon=\epsilon_1+\cdots+\epsilon_{T}$. 
In this work, the privacy budget is evenly split across iterations, so $\epsilon_1=\cdots=\epsilon_{T}=\frac{\epsilon}{T}$. Since Gaussian noise is injected into the embedding matrix $\mathbf{V}$, the $\mathbf{V}$ satisfies $(\frac{\epsilon}{T}, \frac{\delta}{T})$-node-level DP for each iteration. In particular, node privacy is satisfied because the sensitivity is calculated according to the definition of node DP in Section~\ref{resist_sens}. After $n_{epochs}\lfloor{N / \flat}\rfloor$ iterations, $\mathbf{V}$ naturally satisfies $(\epsilon,\delta)$-node-level DP, following the sequential composition property. Also, the resulting graphs obey $(\epsilon,\delta)$-node-level DP, as stipulated by the post-processing property of DP.
\end{proof}
\end{theorem}

% \noindent\textbf{Complexity Analysis.}
\subsection{Complexity Analysis}
Here we analyze the computational complexity of PrivDPR. 
The time complexity for initialization is $O(1)$. 
The time complexity for the outer loop is $n_{epochs}$. 
The time complexity for the inner loop is $\lfloor{N/\flat}\rfloor$. 
Random walk generation has a complexity of $O(\flat R_{wn} R_{wl})$ per batch. 
The complexity of weight normalization and private updates is $O(Ld^2)$. 
The score matrix has a complexity of $O(N^2 r)$. Considering the refined considerations above, the overall time complexity can be expressed as $O(n_{epochs}(\lfloor N / \flat \rfloor (\flat R_{wn} R_{wl} + Ld^2) + N^2r))$. This implies that the time complexity is linear with respect to the number of nodes in the graph, so our method is scalable and can be applied to large-scale graphs.

% The more detailed analysis of the time complexity for PrivDPR and the comparison of different algorithms are presented in \textbf{Appendix~\ref{appendix:time_complex}}.

% Given the refined considerations above, the overall time complexity can be summarized as follows:
% $\mathcal{O}(n_{epochs} \cdot (\lfloor N / \flat \rfloor \cdot (\flat \cdot R_{wn} \cdot R_{wl} + Ld^2) + N^2r))$.

% The overall time complexity of PrivDPR is $\mathcal{O}(n_{epochs}(\lfloor N / \flat \rfloor (\flat R_{wn} R_{wl} + Ld^2) + N^2r))$, which implies that our proposed method can scale effectively to large-scale graphs.
% The more detailed analysis of the time complexity for PrivDPR and the comparison of different algorithms are presented in \textbf{Appendix~\ref{appendix:time_complex}}.
\section{Experiments}\label{Experiment}
In this section, we will answer the following three questions: 
\begin{itemize}
[leftmargin=5mm]
\item How do the weight normalization parameter $s$ and weight dimension $d$ in neural networks affect the performance of PrivDPR? (see Section~\ref{exp:para_impact})
\item How does the privacy budget $\epsilon$ impact on the performance of PrivDPR? (see Section~\ref{exp:Impact_PrivBudget})
\item How scalable is PrivDPR in the context of link prediction and node classification tasks? (see Section~\ref{exp:link_predict})
\end{itemize}

\noindent\textbf{Datasets.}
We run experiments on the five real-world datasets, Cora\footnote{\url{https://linqs.org/datasets/}}, Citeseer~\cite{sen2008collective}, p2p\footnote{\url{http://snap.stanford.edu/data/p2p-Gnutella08.html}}, Chicago\footnote{\url{http://konect.cc/networks/tntp-ChicagoRegional/}}, 
and Amazon\footnote{\url{https://snap.stanford.edu/data/amazon0505.html}}.
Cora is a citation network of academic papers with 2,708 nodes, 7 classes, and 5,429 edges.
Citeseer is a similar citation network with 3,327 nodes, 6 classes, and 4,732 edges.
p2p is a sequence of snapshots from the Gnutella P2P network, consisting of 6,301 nodes and 20,777 edges.
Chicago is a directed transportation network of the Chicago area with 12,982 nodes and 39,018 edges.
Amazon is a co-purchase network with 410,236 nodes and 3,356,824 edges, representing product connections based on co-purchases.
Since we focus on simple graphs in this work, all datasets are pre-processed to remove self-loops. 

\noindent\textbf{Baselines.} 
We compare our PrivDPR~\footnote{Our code is available at \url{https://github.com/sunnerzs/PrivDPR}.} with four other baselines: GAP~\cite{sajadmanesh2023gap}, DPGGAN~\cite{yang2020secure}, DPGVAE~\cite{yang2020secure}, and DPR (No DP). GAP represents the current state-of-the-art differentially private GNN model, designed to produce private node embeddings. 
For a fair comparison, we configure GAP to generate synthetic graphs using the same generation method as PrivDPR. 
For a fair comparison, \emph{we configure GAP to generate synthetic graphs using the same generation method as PrivDPR.}
In this study, we simulate a scenario in which the graphs contain only structural information, whereas GAP depends on node features. To guarantee a fair evaluation, similar to prior research~\cite{du2022understanding}, we employ randomly generated features as inputs for GAP. Given that random features do not infringe on privacy, we eliminate the noise perturbation on features in GAP. DPR (No DP) serves as the non-private version of PrivDPR.

\noindent\textbf{Parameter Settings.}
In PrivDPR, we vary the privacy budget $\epsilon$ from $\{0.1, 0.2, 0.4, 0.8, 1.6, 3.2\}$ while keeping the privacy parameter $\delta$ fixed at $10^{-5}$. The learning rate is set to $\eta = 1 \times 10^{-3}$, which is consistent with the settings used in DPGGAN and DPGVAE. We limit the maximum number of training epochs to $n_{epochs} = 5$. The embedding dimension is chosen as $r = 128$. Note that we do not specifically show the effect of $r$ as its setting is commonly used in various network embedding methods~\cite{perozzi2014deepwalk,tang2015line,lai2017prune,tu2018unified,zhang2018link,du2022understanding}.
The delay factor is set to $\gamma = 0.85$. The batch size is $\flat = 16$. Also, we use $R_{wn} = 2$ random walks with a length of $R_{wl} = 16$. These values are one-fifth of the recommended values in DeepWalk~\cite{perozzi2014deepwalk}.
In Section~\ref{exp:para_impact}, we investigate the impact of adjusting the weight normalization parameter $s$ on PrivDPR, while considering the value of $\mathcal{S}_\nabla$ as 5, which is determined based on DPGGAN. Note that we do not modify the number of layers $L$ in PrivDPR, since it is calculated dynamically based on $s$ using Eq.~(\ref{ResPrivBug}).
Also, in Section~\ref{exp:para_impact}, we assess the influence of varying the dimension of the hidden layer weight $d$ on PrivDPR.
To ensure consistency with the original papers, we employ the official GitHub implementations for GAP, DPGGAN, and DPGVAE. We replicate the experimental setup described in those papers to maintain consistency and comparability.

\noindent\textbf{Graph Utility Metrics.}
To evaluate the similarity between $G$ and $\widetilde{G}$, we use eight graph topological metrics: triangle count (TC), wedge count (WC), claw count (CC), relative edge distribution entropy (REDE), characteristic path length (CPL), Diameter, size of the largest connected component (LCC), and degree distribution. 

We evaluate the accuracy of PrivDPR in the aforementioned graph metrics patterns over all datasets against the baselines. The accuracy of each method $\mathcal{A}$ on graph $G$ is
measured by the mean relative error (MRE)~\cite{sun2019analyzing}, namely 
$\mathrm{MRE} = \frac{1}{|\mathcal{A}|} \sum_{\mathcal{A}_{i} \in \mathcal{A}}\left|\frac{\hat{\mathcal{A}}_{i}(G) -
               \mathcal{A}_{i}(G)}{\mathcal{A}_{i}(G)}\right|$,    
where $\mathcal{A}_{i}(G)$ denotes the true query result
in input graph $G$, and $\hat{\mathcal{A}}_{i}(G)$ denotes the differentially private query result in $G$. Each result reported is averaged over five repeated runs, that is $|\mathcal{A}|=5$. A lower MRE indicates a lower error and thus a higher data utility. 

The degree distribution is measured with Kolmogorov-Smirnov (KS)~\cite{jorgensen2016publishing,jian2021publishing}, which quantifies the maximum distance between the two-degree distributions. 
Let $F$ and $F^\prime$ denote the cumulative distribution functions estimated from the sorted degree sequences of the original and synthetic graphs, respectively. Then $KS_D=\max _d\left|F(d)-F^\prime(d)\right|$.
The smaller this statistic value, the closer (more similar) the degree distributions between the synthetic and original graphs. For KS, we also report the average performance over five independent runs.
     
\subsection{\textbf{Impact of Parameters}}\label{exp:para_impact}
\textbf{Parameter $s$.} In this experiment, we investigate the impact of the parameter $s$ on the performance of PrivDPR, focusing on graph metrics such as TC, REDE, CPL, and KS. For the datasets Cora, Citeseer, p2p, and Chicago, we consider different values of $s$, namely, $2, 4, 6, 8$. As shown in Table~\ref{Tab_WcLayerNum}, we analyze the standard deviation (SD) of the MRE and $KS_D$ values for PrivDPR across various $s$. Remarkably, we consistently observe that the SD remains consistently not more than $3.5122E-02$ across all datasets. This result highlights the robustness of PrivDPR against variations in $s$.
Based on these findings, we have made the decision to set $s$ as a constant value of 8 in subsequent experiments.

\noindent\textbf{Parameter $d$.} In this experiment, we investigate the impact of the weight dimension parameter $d$ on the performance of PrivDPR in terms of various graph metrics, including TC, REDE, CPL, and $KS_D$. Specifically, we consider different values of $d$, namely, $64, 128, 256, 512$ for datasets Cora, Citeseer, p2p, and Chicago.
As illustrated in Table~\ref{Tab_Wdim}, we observe that the SD of the MRE and $KS_D$ values for PrivDPR with different $d$ is consistently not more than $5.8030E-02$ across all datasets. This finding indicates that PrivDPR exhibits robustness to variations in $d$. Consequently, in subsequent experiments, we fix $d$ at $64$ as a constant value.

\begin{table}[htb!]\small
\caption{Summary of MRE and $KS_D$ with different $s$, given $\epsilon=3.2$ and $\mathcal{S}_\nabla=5$. SD represents the standard deviation of the values in each row.}
\label{Tab_WcLayerNum}
\begin{tabular}{c|c|cccc|c}
\hline
\multirow{2}{*}{Dataset}  & \multirow{2}{*}{Statistics} & \multicolumn{4}{c|}{Parameter $s$}           & \multirow{2}{*}{\textbf{SD}} \\ \cline{3-6}
                          &                             & 2 & 4 & 6 & 8 &                               \\ \hline
\multirow{4}{*}{Cora}     & TC                          & 0.9926  & 0.9914  & 0.9908  & 0.9893  & \textbf{1.4029E-03}           \\
                          & REDE                        & 0.0244  & 0.0245  & 0.0242  & 0.0245  & \textbf{1.6113E-04}           \\
                          & CPL                         & 0.1135  & 0.1106  & 0.1155  & 0.1162  & \textbf{2.4788E-03}           \\
                          & KS                          & 0.5231  & 0.5223  & 0.5224  & 0.5356  & \textbf{6.4809E-03}           \\ \hline
\multirow{4}{*}{Citeseer} & TC                          & 0.9923  & 0.9944  & 0.9919  & 0.9936  & \textbf{1.1799E-03}           \\
                          & REDE                        & 0.0168  & 0.0182  & 0.0172  & 0.0165  & \textbf{7.5261E-04}           \\
                          & CPL                         & 0.3105  & 0.3068  & 0.3116  & 0.3198  & \textbf{5.4899E-03}           \\
                          & KS                          & 0.5792  & 0.5249  & 0.5830  & 0.6084  & \textbf{3.5122E-02}           \\ \hline
\multirow{4}{*}{p2p}      & TC                          & 0.9591  & 0.9593  & 0.9566  & 0.9557  & \textbf{1.7917E-03}           \\
                          & REDE                        & 0.0381  & 0.0379  & 0.0382  & 0.0382  & \textbf{1.5189E-04}           \\
                          & CPL                         & 0.0247  & 0.0232  & 0.0227  & 0.0227  & \textbf{9.5005E-04}           \\
                          & KS                          & 0.1322  & 0.1325  & 0.1224  & 0.1177  & \textbf{7.3605E-03}           \\ \hline
\multirow{4}{*}{Chicago}  & TC                          & 0.9913  & 0.9901  & 0.9950  & 0.9938  & \textbf{2.2624E-03}           \\
                          & REDE                        & 0.0078  & 0.0076  & 0.0078  & 0.0080  & \textbf{1.7852E-04}           \\
                          & CPL                         & 0.8118  & 0.8114  & 0.8124  & 0.8121  & \textbf{4.3182E-04}           \\
                          & KS                          & 0.2854  & 0.2854  & 0.2854  & 0.2778  & \textbf{3.8314E-03}           \\ \hline
\end{tabular}
\end{table}

\begin{table}[htb!]\small
\caption{Summary of MRE and $KS_D$ with different $d$, given $\epsilon=3.2$ and $s=8$. SD represents the standard deviation of the values in each row.}
\label{Tab_Wdim}
\begin{tabular}{c|c|llll|l}
\hline
\multirow{2}{*}{Dataset}  & \multirow{2}{*}{Statistics} & \multicolumn{4}{c|}{Weight Dimension $d$}                                                                   & \multicolumn{1}{c}{\multirow{2}{*}{\textbf{SD}}} \\ \cline{3-6}
                          &                             & \multicolumn{1}{c}{64} & \multicolumn{1}{c}{128} & \multicolumn{1}{c}{256} & \multicolumn{1}{c|}{512} & \multicolumn{1}{c}{}                              \\ \hline
\multirow{4}{*}{Cora}     & TC                          & 0.9893                 & 0.9920                  & 0.9905                  & 0.9914                   & \textbf{1.2012E-03}                               \\
                          & REDE                        & 0.0245                 & 0.0252                  & 0.0248                  & 0.0234                   & \textbf{7.8200E-04}                               \\
                          & CPL                         & 0.1162                 & 0.1068                  & 0.1101                  & 0.1174                   & \textbf{5.0061E-03}                               \\
                          & KS                          & 0.5356                 & 0.5216                  & 0.5227                  & 0.5388                   & \textbf{8.7785E-03}                               \\ \hline
\multirow{4}{*}{Citeseer} & TC                          & 0.9936                 & 0.9923                  & 0.9944                  & 0.9936                   & \textbf{8.8327E-04}                               \\
                          & REDE                        & 0.0165                 & 0.0172                  & 0.0179                  & 0.0170                   & \textbf{5.7647E-04}                               \\
                          & CPL                         & 0.3198                 & 0.3128                  & 0.3041                  & 0.3159                   & \textbf{6.6596E-03}                               \\
                          & KS                          & 0.6084                 & 0.5507                  & 0.4888                  & 0.6120                   & \textbf{5.8030E-02}                               \\ \hline
\multirow{4}{*}{p2p}      & TC                          & 0.9557                 & 0.9532                  & 0.9547                  & 0.9616                   & \textbf{3.6798E-03}                               \\
                          & REDE                        & 0.0382                 & 0.0381                  & 0.0383                  & 0.0383                   & \textbf{1.0559E-04}                               \\
                          & CPL                         & 0.0227                 & 0.0236                  & 0.0205                  & 0.0226                   & \textbf{1.3138E-03}                               \\
                          & KS                          & 0.1177                 & 0.1572                  & 0.1276                  & 0.1271                   & \textbf{1.7174E-02}                               \\ \hline
\multirow{4}{*}{Chicago}  & TC                          & 0.9938                 & 0.9926                  & 0.9857                  & 0.9932                   & \textbf{3.7517E-03}                               \\
                          & REDE                        & 0.0080                 & 0.0077                  & 0.0077                  & 0.0077                   & \textbf{1.5630E-04}                               \\
                          & CPL                         & 0.8121                 & 0.8114                  & 0.8115                  & 0.8116                   & \textbf{3.2051E-04}                               \\
                          & KS                          & 0.2778                 & 0.2689                  & 0.3032                  & 0.2854                   & \textbf{1.4568E-02}                               \\ \hline
\end{tabular}
\end{table}

\subsection{Impact of Privacy Budget on Graph Statistics Preservation}\label{exp:Impact_PrivBudget}
The privacy budget $\epsilon$ is a critical parameter in the context of DP, as it determines the level of privacy provided by the algorithm. We conduct experiments to evaluate the impact of $\epsilon$ on the performance of each private algorithm. Specifically, we present the results on Cora in Figure~\ref{PrivBud_on_cora}. The results on Citeseer, p2p and Chicago datasets can be found in \textbf{Appendix~\ref{appendix_exp:p2p_chicago}}.  
From these figures, PrivDPR consistently outperforms both DPGGAN and DPGVAE with weak edge-level DP guarantees. The reasons why the results of DPGGAN and DPGVAE are poor are twofold: 1) they use the MA mechanism and tend to converge prematurely under small $\epsilon$; and 2) they use a threshold-based method for reconstructing synthetic graphs, which potentially generates numerous disconnected subgraphs within the synthesized graph.
Moreover, it is worth noting that in most cases, PrivDPR achieves comparable results to DPR (No DP) and surpasses GAP, even with a small privacy budget of $\epsilon=0.1$. This phenomenon can be attributed to two factors. 
First, our designed deep PageRank effectively captures the structural properties of the input graph. 
Second, the theorem presented in Section~\ref{sec:gra_perturb} addresses challenges such as high sensitivity and excessive splitting on the privacy budget, further enhancing the performance of PrivDPR.

\begin{figure*}[htb!]
%  \centering
  \centerline{\includegraphics[width=6.7in]{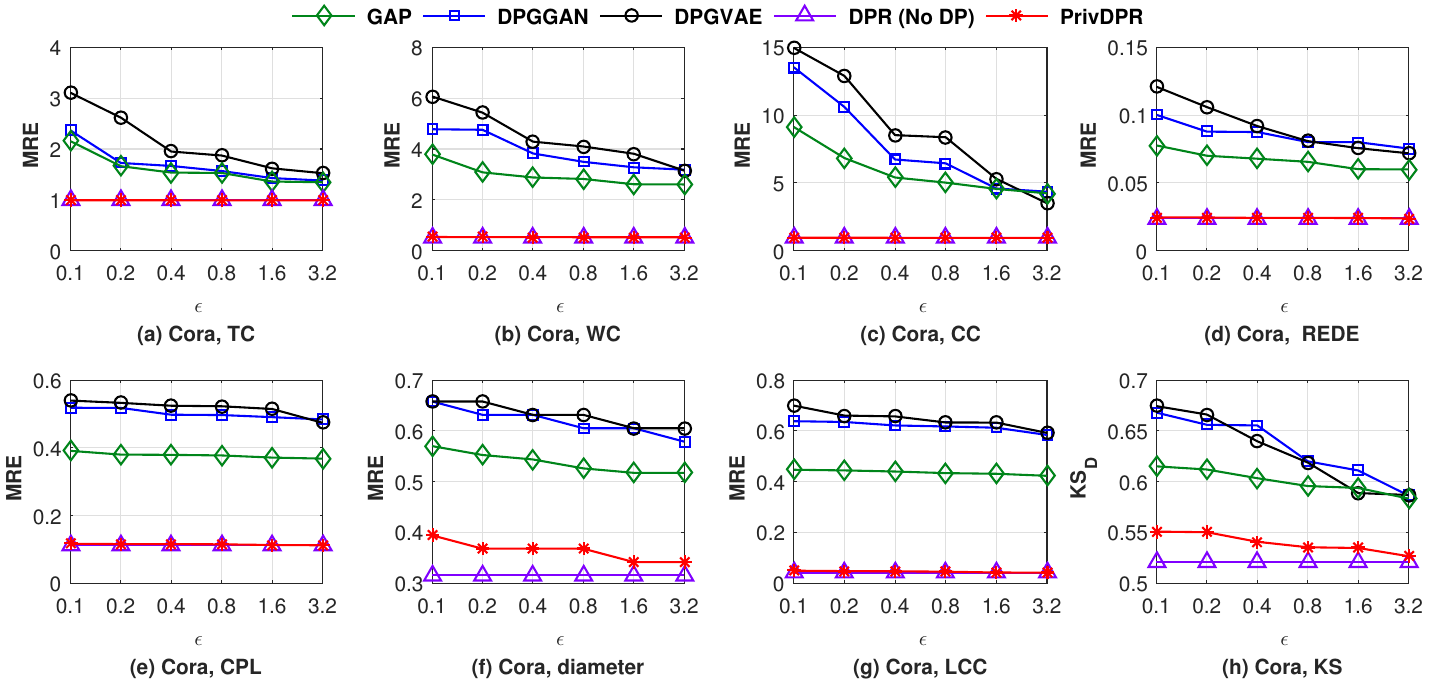}}
  \caption{Privacy budget on Cora}\label{PrivBud_on_cora}
\end{figure*}

\subsection{Link Prediction and Node Classification}\label{exp:link_predict}
For the link prediction task, the existing links in each dataset are randomly divided into a training set (80\%) and a test set (20\%). To evaluate the performance of link prediction, we randomly select an equal number of node pairs without connected edges as negative test links for the test set. Additionally, for the training set, we sample the same number of node pairs without edges to construct negative training data. We measure performance using the area under the ROC curve (AUC). The AUC results and analysis for all methods are presented in Tables~\ref{tab:auc_score1} and \ref{tab:auc_score}, with $\epsilon=0.1$ and $\epsilon=3.2$, in \textbf{Appendix~\ref{appendix_exp:link_pred}}. In summary, PrivDPR consistently achieves the highest AUC among all privacy-preserving algorithms and maintains high stability.

For the node classification task, we randomly sample 90\% of the nodes as training data and randomly sample 10\% of the nodes outside the training set as test data. We follow the procedure of~\cite{epasto2022differentially} and evaluate our embeddings using Micro-F1 score.
We report the results in Table~\ref{tab:micro_f1} in \textbf{Appendix~\ref{appendix_exp:node_class}}, with $\epsilon=0.1$ and $\epsilon=3.2$. To summarize, PrivDPR outperforms other privacy methods across various datasets in terms of Micro-F1 and SD. This indicates PrivDPR is an outstanding off-the-shelf method for different graph downstream tasks.

\section{Related Work}\label{Related_work}
Unlike traditional privacy-preserving methods~\cite{sweeney2002k,machanavajjhala2007diversity,hu2014private}, differential privacy (DP) and its variant, local differential privacy (LDP), offer strong privacy guarantees and robustness against adversaries with prior knowledge and have been widely applied across various fields~\cite{ye2019privkv, ye2021privkvm, ye2021beyond, ye2023stateful, Zhang2023trajectory}. The related work of this paper covers differentially private shallow and deep graph generation models.

\subsection{Private Shallow Graph Generation Models}
Several research efforts have been dedicated to achieving differentially private publication for social graph data.
One approach is to generate representative synthetic graphs using the Kronecker graph model, as explored by Mir and Wright~\cite{Mir2009ADP}. They estimate the model parameters from the input graph under DP.
Another approach is the Pygmalion model proposed by Sala {\it et al}., which utilizes the dK-series of the input graph to capture the distribution of observed degree pairs on edges~\cite{sala2011sharing}. This model has been combined with smooth sensitivity to construct synthetic graphs~\cite{wang2013preserving}.
Xiao {\it et al}. encode the graph structure through private edge counting queries under the hierarchical random graph model and report improved results compared to the dK-series approach~\cite{xiao2014differentially}.
Chen {\it et al}. employ the exponential mechanism to sample an adjacency matrix after clustering the input graph~\cite{chen2014correlated}. Proserpio {\it et al}. suggest down-weighting the edges of a graph non-uniformly to mitigate high global sensitivity arising from very high degree nodes. They demonstrate this approach, combined with MCMC-based sampling, for generating private synthetic graphs~\cite{proserpio2014calibrating}.
Gao and Li propose a private scheme to preserve both the adjacency matrix and persistent homology, specifically targeting the persistence structures in the form of holes~\cite{gao2019phdp}. 
These methods, along with subsequent research on private graph release~\cite{eliavs2020differentially, gupta2012iterative, yuan2023privgraph}, typically guarantee weak edge-level DP.

\subsection{Private Deep Graph Generation Models}
With the rapid development of deep learning, numerous advanced deep graph generation models have emerged in recent times. These models employ various powerful neural networks in a learn-to-generate manner~\cite{kipf2016variational,bojchevski2018netgan,you2018graphrnn}. For instance, NetGAN~\cite{bojchevski2018netgan} converts graphs into biased random walks, learns to generate walks using GAN, and then constructs graphs from the generated walks. GraphRNN~\cite{you2018graphrnn}, on the other hand, treats graph generation as a sequence of node and edge additions, and models it using a heuristic breadth-first search scheme and hierarchical RNN. These deep learning models can generate graphs with richer properties and flexible structures learned from real-world networks. However, existing research on deep graph generation has not thoroughly examined the potential privacy threats associated with training and generating graphs using powerful models. A recent related solution proposed by Yang {\it et al}. introduces DPGGAN and DPGVAE models for graph synthesis~\cite{yang2020secure}. They improve MA~\cite{abadi2016deep}, which is an effective strategy for computing privacy loss after multiple queries, but only achieves weak edge-level DP. Recent research endeavors~\cite{daigavane2021node,zhang2024dpar,sajadmanesh2023gap,xiang2023preserving} have been dedicated to the advancement of node-level differentially private GNNs. These efforts aim to address the issue of high sensitivity, but achieving an optimal balance between privacy and utility remains a significant challenge. This is primarily because they often employ advanced composition mechanisms to manage privacy budget splitting and yet tend to converge prematurely when working with a small privacy budget. 
\section{Conclusion}\label{ConcluSection}
This paper focuses on synthetic graph generation under DP. The underlying highlights lie in the following two aspects. 
First, we design a novel privacy-preserving deep PageRank for graph synthesis, called PrivDPR, which achieves DP by adding noise to the gradient for a specific weight during learning.
Second, we theoretically show that increasing the number of layers can effectively overcome the challenges associated with high sensitivity and privacy budget splitting.
Through privacy analysis, we prove that the generated synthetic graph satisfies $(\epsilon,\delta)$-node-level DP. Extensive experiments on real-world graph datasets show that our solution substantially outperforms state-of-the-art competitors. 
Our future focus is on developing more graph generation techniques for node embeddings to enhance the utility of synthetic graphs.

\begin{acks}
This work was supported by the National Natural Science Foundation of China (Grant No: 62372122, 92270123 and 62072390), and the Research Grants Council, Hong Kong SAR, China (Grant No: 15224124, 25207224, C2004-21GF and C2003-23Y).
We also thank Dr. Nan Fu and Dr. Lihe Hou for helpful discussions.
\end{acks}

\bibliographystyle{ACM-Reference-Format}
\bibliography{mybibfile.bib}
% \newpage
\begin{appendices}
\section{Graph Synthesis}\label{appendix:gra_syn}
In what follows, we describe the specific details of generating graphs as discussed in previous works~\cite{bojchevski2018netgan, rendsburg2020netgan}. Once training is complete, we employ the node embeddings to create a score matrix $\mathbf{S}$ that records transition counts. Since we intend to analyze this synthetic graph, we convert the raw counts matrix $\mathbf{S}$ into a binary adjacency matrix. Initially, $\mathbf{S}$ is symmetrized by setting $s_{ij} = s_{ji} = \max(s_{ij}, s_{ji})$. However, since we have no explicit control over the starting node of the random walks generated by $G$, high-degree nodes are likely to be overrepresented. Consequently, a simple strategy such as thresholding or selecting the top-$k$ entries for binarization may exclude low-degree nodes and create isolated nodes. To address this concern, we ensure that every node $i$ has at least one edge by sampling a neighbor $j$ with a probability of $p_{ij} = \frac{s_{ij}}{\sum_v s_{iv}}$. If an edge has already been sampled, we repeat this procedure. To ensure the graph is undirected, we include $(j,i)$ for every edge $(i,j)$. We continue sampling edges without replacement, using the probability $p_{ij} = \frac{s_{ij}}{\sum_{u, v} s_{uv}}$ for each edge $(i,j)$, until we reach the desired number of edges (e.g., determined by applying the Sigmoid function with the threshold value 0.5 to the score matrix). Additionally, to ensure the graph is undirected, we include $(j,i)$ for every edge $(i,j)$.   

\section{Proof of Lemma~\ref{lemma_uppbound_objfunc}}\label{appendix:lemma_proof}
\begin{proof}
Let $P_j$ denote the set of direct predecessors of node $j$, and $d_i^{out}$ represent the out-degree of node $i$. According to Eq.~(\ref{DPRObjFunc}), the objective function can be expressed as follows:
 \begin{equation*}
 \begin{split}
       \min_{\boldsymbol{\Theta}} \mathscr{L}
     =& \sum_{j \in V}\gamma^2\left(\sum_{i \in P_j} \frac{f(v_i; \boldsymbol{\Theta})}{d_i^{out}}+\frac{1-\gamma}{\gamma N} - \frac{f(v_j; \boldsymbol{\Theta})}{\gamma}\right)^2 \\
     \stackrel{(1)}=& \sum_{j \in V}\gamma^2\left(\sum_{i \in P_j} \frac{f(v_i; \boldsymbol{\Theta})}{d_i^{out}}- \frac{f(v_j; \boldsymbol{\Theta})}{\gamma}\right)^2 \\
     &+ \sum_{j \in V}2\gamma^2\left(\sum_{i \in P_j} \frac{f(v_i; \boldsymbol{\Theta})}{d_i^{out}}- \frac{f(v_j; \boldsymbol{\Theta})}{\gamma}\right)\frac{1-\gamma}{\gamma N}
    +\sum_{j \in V}\frac{(1-\gamma)^2}{N^2} \\
    \stackrel{(2)}\leq& \sum_{(i,j) \in E}d_j^{in}\gamma^2\left(\frac{f(v_i; \boldsymbol{\Theta})}{d_i^{out}}- \frac{f(v_j; \boldsymbol{\Theta})}{d_j^{in}\gamma}\right)^2 \\
    & +\sum_{(i,j) \in E}\left(\frac{f(v_i; \boldsymbol{\Theta})}{d_i^{out}}- \frac{f(v_j; \boldsymbol{\Theta})}{d_j^{in}\gamma}\right)\frac{2\gamma(1-\gamma)}{N}
    +\sum_{(i,j) \in E}\frac{(1-\gamma)^2}{d_j^{in}N^2},
 \end{split}
\end{equation*}
where the step (2) holds because the Cauchy-Schwarz inequality is applied to the first term in the step (1).
\end{proof}

\section{Proof of Lemma~\ref{lemma_NNObj_On_V}}\label{appendix:lemma2_proof}
\begin{proof}
The Lipschitz constant of a function and the norm of its gradient are two sides of the same coin. We define $\|f\|_\mathrm{Lip}$ as the smallest value $\rho$ such that $\|f(\mathbf{x})-f(\mathbf{x}^\prime)\| / \|\mathbf{x}-\mathbf{x}^\prime\| \leq \rho$ for any $\mathbf{x}, \mathbf{x}^\prime$, with the norm being the $\ell_2$-norm. We can use the inequality $\|f_1 \circ f_2\|_\mathrm{Lip}\leq\|f_1\|_\mathrm{Lip}\cdot\|f_2\|_\mathrm{Lip}$ to observe the following bound on $\|f\|_\mathrm{Lip}$:
\begin{equation*}\label{NeuralNetGra_On_V}
    \begin{split}
         \left\|f\right\|_\mathrm{Lip}        \leq&\left\|\phi_{L+1}\right\|_\mathrm{Lip}\cdot\left\|\mathbf{W}_{L+1}\mathbf{V}_L\right\|_\mathrm{Lip}\cdot\left\|\phi_L\right\|_\mathrm{Lip}\cdot\left\|\mathbf{W}_L\mathbf{V}_{L-1}\right\|_\mathrm{Lip} \\ &\cdots\left\|\phi_1\right\|_\mathrm{Lip}\cdot\left\|\mathbf{W}_1\mathbf{V}_0\right\|_\mathrm{Lip} 
       =\prod_{l=1}^{L+1}\left\|\phi_l\right\|_\mathrm{Lip} \cdot \left\|\mathbf{W}_l\mathbf{V}_{l-1}\right\|_\mathrm{Lip} \\
       \stackrel{(i)}\leq& \prod_{l=1}^{L+1} \left\|\mathbf{W}_l\mathbf{V}_{l-1}\right\|_\mathrm{Lip}
       \stackrel{(ii)}=\prod_{l=1}^{L+1}\max_{\mathbf{V} \neq 0}\frac{\left\|\mathbf{W}_l\mathbf{V}\right\|_2}{\left\|\mathbf{V}\right\|_2}=\prod_{l=1}^{L+1}\left\|\mathbf{W}_l\right\|_2,
    \end{split}
\end{equation*}
where the step (\emph{i}) holds since the \emph{Sigmoid} function is used, and the step (\emph{ii}) holds because weight normalization is implemented through spectral normalization. The \emph{Sigmoid} function can be substituted with other activation functions, such as \emph{ReLU} and \emph{leaky ReLU}, while the inequalities above still hold.

According to $\|f\|_\mathrm{Lip}$, it is possible to show that $\left\|\frac{\partial f(v; \boldsymbol{\Theta})}{\partial\mathbf{V}}\right\|_2 \leq \prod_{l=1}^{L+1} \left\|\mathbf{W}_l\right\|_2$.
This fact allows us to bound the weight $\mathbf{W}_l$, instead of the output of the gradient function, and obtain a bound on the $\ell_2$-sensitivity of the gradient.
\end{proof}

\section{Proof of Lemma~\ref{lemma_BatchObjFuncGra_On_V}}\label{appendix:lemma3_proof}
\begin{proof}
With Lemma~\ref{lemma_NNObj_On_V}, we have
\begin{equation*}\label{GraUpBound}
    \begin{split}
        &\left\|\frac{\partial{\mathcal{L}(v_i, v_j; \boldsymbol{\Theta})}}{\partial{\mathbf{V}}}\right\|_2
      =\left\|\left(\frac{2d_j^{in}\gamma^2f(v_i; \boldsymbol{\Theta})}{d_i^{out}}- 2\gamma f(v_j; \boldsymbol{\Theta}) + \frac{2\gamma(1-\gamma)}{N}\right) \right.\\
           &\left. \left(\frac{\partial f(v_i; \boldsymbol{\Theta})/\partial\mathbf{V}}{d_i^{out}}- \frac{\partial f(v_j; \boldsymbol{\Theta})/\partial\mathbf{V}}{d_j^{in}\gamma}\right)\right\|_2 \\
      \stackrel{(1)}
      \leq& \left\|\left(\frac{2d_j^{in}\gamma^2}{d_i^{out}}+ 2\gamma +
      \frac{2\gamma(1-\gamma)}{N}\right) \left(\frac{\partial f(v_i; \boldsymbol{\Theta})/\partial\mathbf{v}_i}{d_i^{out}}- \frac{\partial f(v_j; \boldsymbol{\Theta})/\partial\mathbf{v}_j}{d_j^{in}\gamma}\right)\right\|_2 \\
      \leq& \left(\frac{2d_j^{in}\gamma^2}{d_i^{out}}+ 2\gamma +
      \frac{2\gamma(1-\gamma)}{N}\right)
      \left(\frac{1}{d_i^{out}} \left\|\prod_{l=1}^{L+1} \mathbf{W}_l\right\|_2 + \frac{1}{d_j^{in}\gamma} \left\|\prod_{l=1}^{L+1} \mathbf{W}_l\right\|_2\right) \\
      \leq& \left(2(N-1)\gamma^2 + 2\gamma + \frac{2\gamma(1-\gamma)}{N}\right)
          \left(1 + \frac{1}{\gamma}\right)\prod_{l=1}^{L+1}\left\|\mathbf{W}_l\right\|_2,
    \end{split}
\end{equation*}
where the step (1) holds since the $f(v; \boldsymbol{\Theta})$ with \emph{Sigmoid} activation in Eq.~(\ref{NeuralNet_model}) is less than 1, and triangle inequality is used.
\end{proof}

\section{Graph Statistics}\label{appendix_exp:p2p_chicago}
The findings for link prediction across the Citeseer, p2p and Chicago datasets are displayed in Figures~\ref{PrivBud_on_p2p} and \ref{PrivBud_on_chicago}.

\begin{figure}[htb!]
  \includegraphics[width=3.5in]{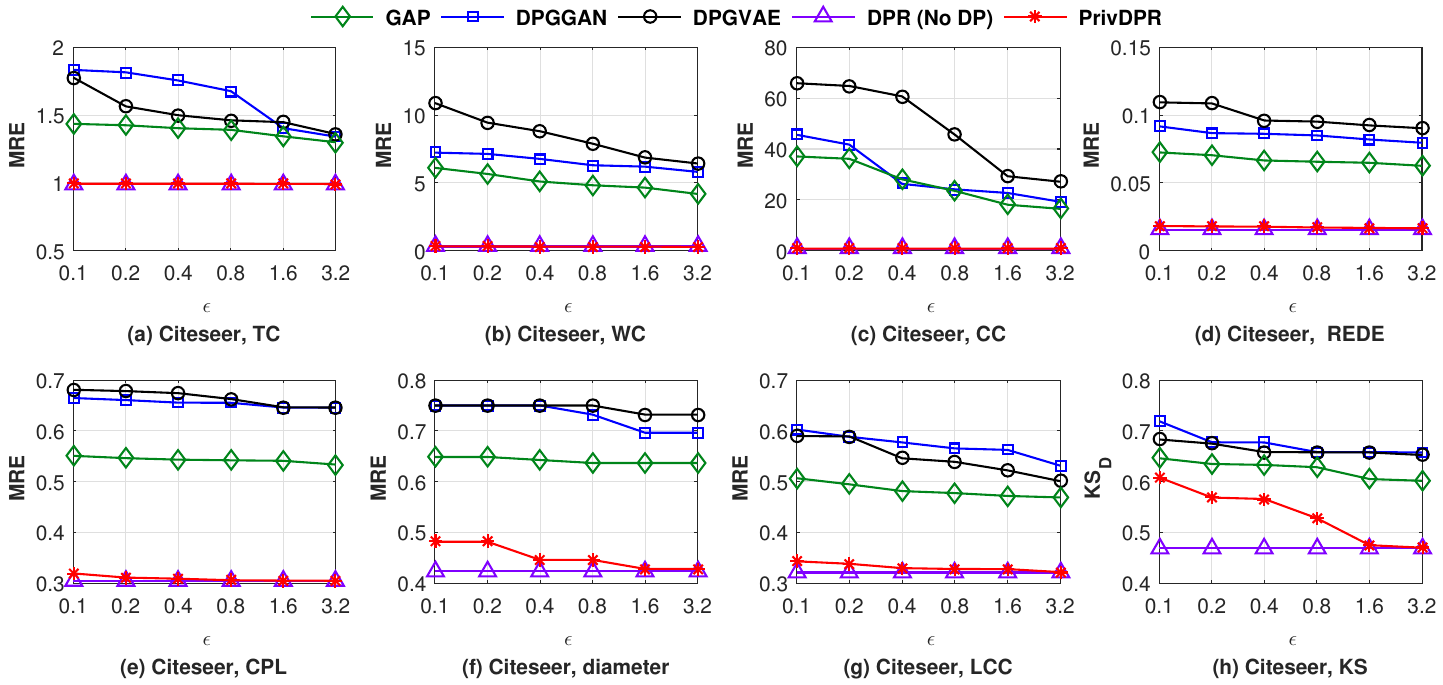}
  \caption{Privacy budget on Citeseer}\label{PrivBud_on_citeseer}
\end{figure}

\begin{figure}[htb!]
  \includegraphics[width=3.5in]{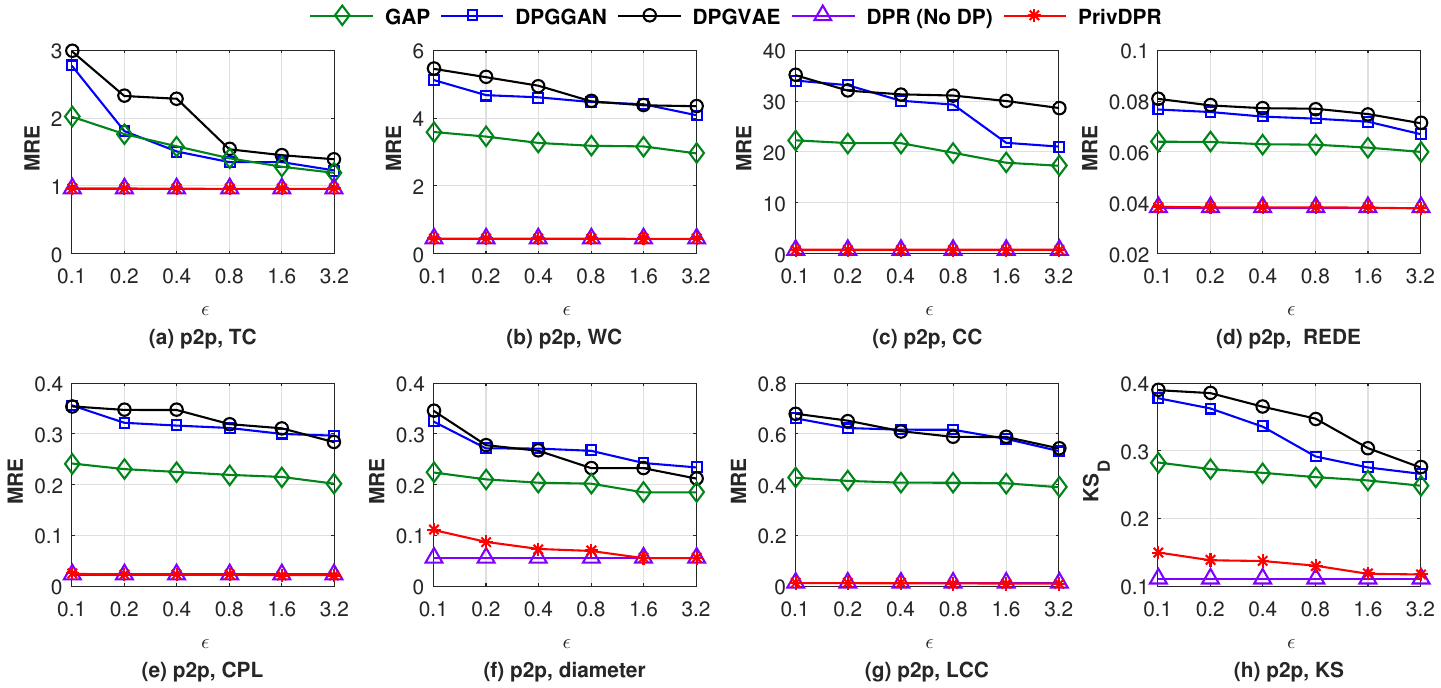}
  \caption{Privacy budget on p2p}\label{PrivBud_on_p2p}
\end{figure}

\begin{figure}[htb!]
  \includegraphics[width=3.5in]{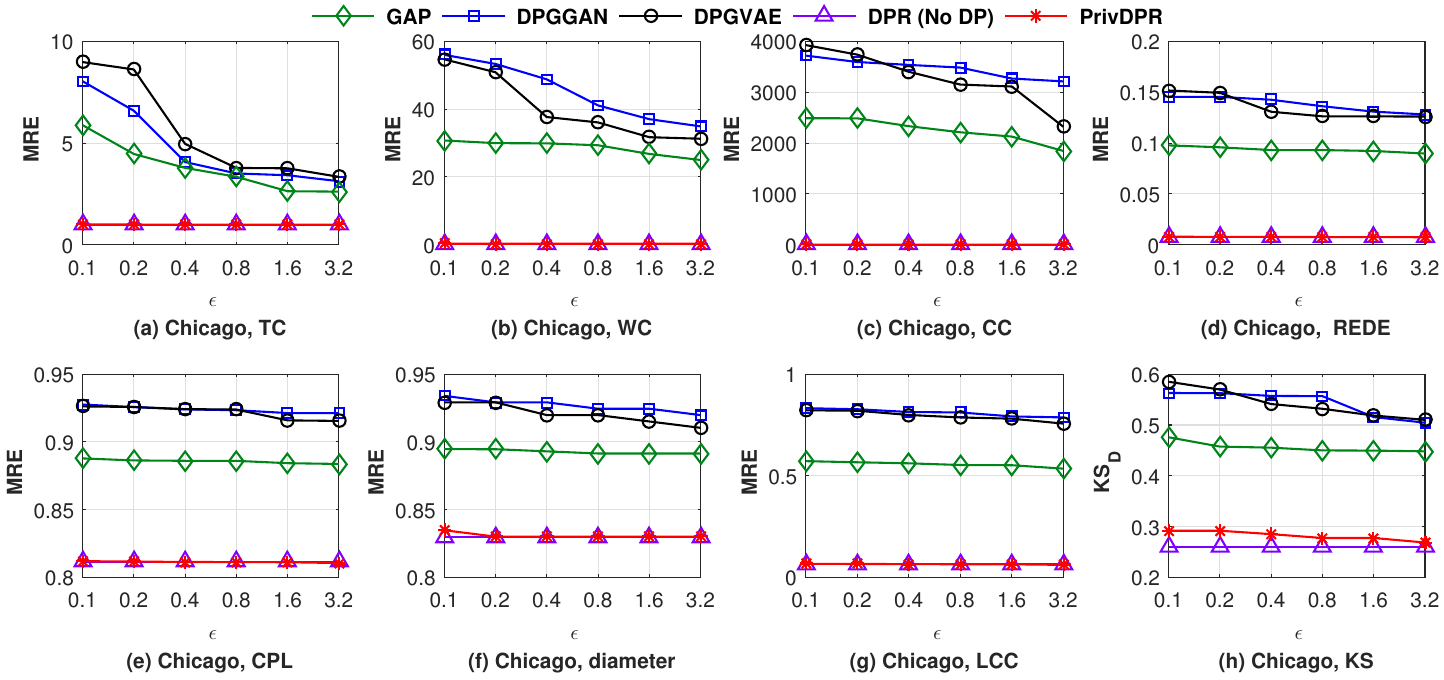}
  \caption{Privacy budget on Chicago}\label{PrivBud_on_chicago}
\end{figure}

\begin{table}[htb!]\small
\caption{AUC Scores for Link Prediction under $\epsilon=0.1$. \textbf{Bold}: best}\label{tab:auc_score1}
\begin{tabular}{c|c|ccccc}
\hline
\multirow{2}{*}{Algorithms} & \multirow{2}{*}{AUC} & \multicolumn{5}{c}{Datasets}                                                                                                                                                \\ \cline{3-7} 
                            &                      & \multicolumn{1}{c|}{Cora}            & \multicolumn{1}{c|}{Citeseer}        & \multicolumn{1}{c|}{p2p}             & \multicolumn{1}{c|}{Chicago}         & Amazon          \\ \hline
\multirow{2}{*}{GAP}        & Mean                 & \multicolumn{1}{c|}{{0.4873}}          & \multicolumn{1}{c|}{{0.4781}}          & \multicolumn{1}{c|}{{0.4859}}          & \multicolumn{1}{c|}{{0.4915}}          & {{0.4792}}          \\ \cline{2-7} 
                            & SD                   & \multicolumn{1}{c|}{{0.0371}}          & \multicolumn{1}{c|}{{0.0213}}          & \multicolumn{1}{c|}{{0.0092}}          & \multicolumn{1}{c|}{{0.0093}} & {{0.0127}}          \\ \hline
\multirow{2}{*}{DPGGAN}     & Mean                 & \multicolumn{1}{c|}{0.4914}          & \multicolumn{1}{c|}{0.4926}          & \multicolumn{1}{c|}{0.5005}          & \multicolumn{1}{c|}{0.4980}          & {0.4978}          \\ \cline{2-7} 
                            & SD                   & \multicolumn{1}{c|}{0.0157}          & \multicolumn{1}{c|}{0.0108} & \multicolumn{1}{c|}{0.0115}          & \multicolumn{1}{c|}{0.0091}          & {0.0058}          \\ \hline
\multirow{2}{*}{DPGVAE}     & Mean                 & \multicolumn{1}{c|}{0.4893}          & \multicolumn{1}{c|}{0.4977}          & \multicolumn{1}{c|}{0.4982}          & \multicolumn{1}{c|}{0.5014}          & {0.4813}          \\ \cline{2-7} 
                            & SD                   & \multicolumn{1}{c|}{0.0083}          & \multicolumn{1}{c|}{0.0156}          & \multicolumn{1}{c|}{0.0091}          & \multicolumn{1}{c|}{0.0064}          & {0.0063}          \\ \hline
\multirow{2}{*}{PrivDPR}    & Mean                 & \multicolumn{1}{c|}{\textbf{0.5066}} & \multicolumn{1}{c|}{\textbf{0.5099}} & \multicolumn{1}{c|}{\textbf{0.5084}} & \multicolumn{1}{c|}{\textbf{0.5020}} & {\textbf{0.5009}} \\ \cline{2-7} 
                            & SD                   & \multicolumn{1}{c|}{\textbf{0.0070}} & \multicolumn{1}{c|}{{\textbf{0.0100}}}          & \multicolumn{1}{c|}{\textbf{0.0052}} & \multicolumn{1}{c|}{\textbf{0.0054}}          & {\textbf{0.0025}} \\ \hline
\end{tabular}
\end{table}

\begin{table}[htb!]\small
\caption{AUC Scores for Link Prediction under $\epsilon=3.2$. \textbf{Bold}: best}\label{tab:auc_score}
\begin{tabular}{c|c|ccccc}
\hline
\multirow{2}{*}{Algorithms} & \multirow{2}{*}{AUC} & \multicolumn{5}{c}{Datasets}                                                                                                                                                \\ \cline{3-7} 
                            &                      & \multicolumn{1}{c|}{Cora}            & \multicolumn{1}{c|}{Citeseer}        & \multicolumn{1}{c|}{p2p}             & \multicolumn{1}{c|}{Chicago}         & Amazon          \\ \hline
\multirow{2}{*}{GAP}        & Mean                 & \multicolumn{1}{c|}{{0.4978}}          & \multicolumn{1}{c|}{{0.4897}}          & \multicolumn{1}{c|}{{0.4973}}          & \multicolumn{1}{c|}{{0.5001}}          & {0.4913}          \\ \cline{2-7} 
                            & SD                   & \multicolumn{1}{c|}{{0.0138}}          & \multicolumn{1}{c|}{{0.0173}}          & \multicolumn{1}{c|}{{0.0214}}          & \multicolumn{1}{c|}{{\textbf{0.0078}}} & {0.0071}          \\ \hline
\multirow{2}{*}{DPGGAN}     & Mean                 & \multicolumn{1}{c|}{0.5189}          & \multicolumn{1}{c|}{0.5027}          & \multicolumn{1}{c|}{0.5079}          & \multicolumn{1}{c|}{0.5031}          & 0.5041          \\ \cline{2-7} 
                            & SD                   & \multicolumn{1}{c|}{0.0136}          & \multicolumn{1}{c|}{\textbf{0.0102}} & \multicolumn{1}{c|}{0.0060}          & \multicolumn{1}{c|}{0.0112}          & 0.0040          \\ \hline
\multirow{2}{*}{DPGVAE}     & Mean                 & \multicolumn{1}{c|}{0.5068}          & \multicolumn{1}{c|}{0.5139}          & \multicolumn{1}{c|}{0.5045}          & \multicolumn{1}{c|}{0.5139}          & 0.5012          \\ \cline{2-7} 
                            & SD                   & \multicolumn{1}{c|}{0.0063}          & \multicolumn{1}{c|}{0.0152}          & \multicolumn{1}{c|}{\textbf{0.0035}}          & \multicolumn{1}{c|}{0.0152}          & 0.0017          \\ \hline
\multirow{2}{*}{PrivDPR}    & Mean                 & \multicolumn{1}{c|}{\textbf{0.5535}} & \multicolumn{1}{c|}{\textbf{0.5658}} & \multicolumn{1}{c|}{\textbf{0.5986}} & \multicolumn{1}{c|}{\textbf{0.5738}} & \textbf{0.5917} \\ \cline{2-7} 
                            & SD                   & \multicolumn{1}{c|}{\textbf{0.0057}} & \multicolumn{1}{c|}{0.0110}          & \multicolumn{1}{c|}{0.0052} & \multicolumn{1}{c|}{0.0094}          & \textbf{0.0013} \\ \hline
\end{tabular}
\end{table}

\begin{table}[htb!]\small
\caption{Micro-F1 Scores for Node Classification under $\epsilon=0.1$ and $\epsilon=3.2$. \textbf{Bold}: best}
\label{tab:micro_f1}
\begin{tabular}{c|c|cc|cc}
\hline
\multirow{2}{*}{Algorithms} & \multirow{2}{*}{Micro-F1} & \multicolumn{2}{c|}{$\epsilon=0.1$}    & \multicolumn{2}{c}{$\epsilon=3.2$}     \\ \cline{3-6} 
                            &                      & \multicolumn{1}{c|}{Cora}   & Citeseer & \multicolumn{1}{c|}{Cora}   & Citeseer \\ \hline
\multirow{2}{*}{GAP}        & Mean                 & \multicolumn{1}{c|}{0.3097} & 0.2094   & \multicolumn{1}{c|}{0.3151} & 0.2238   \\ \cline{2-6} 
                            & SD                   & \multicolumn{1}{c|}{0.0250} & \textbf{0.0010}   & \multicolumn{1}{c|}{\textbf{0.0011}} & 0.0150   \\ \hline
\multirow{2}{*}{DPGGAN}     & Mean                 & \multicolumn{1}{c|}{0.2690} & 0.1855   & \multicolumn{1}{c|}{0.2849} & 0.2042   \\ \cline{2-6} 
                            & SD                   & \multicolumn{1}{c|}{0.0150} & 0.0161   & \multicolumn{1}{c|}{0.0096} & 0.0080   \\ \hline
\multirow{2}{*}{DPGVAE}     & Mean                 & \multicolumn{1}{c|}{0.2915} & 0.2075   & \multicolumn{1}{c|}{0.3026} & 0.2172   \\ \cline{2-6} 
                            & SD                   & \multicolumn{1}{c|}{0.0088} & 0.0149   & \multicolumn{1}{c|}{0.0148} & \textbf{0.0030}   \\ \hline
\multirow{2}{*}{PrivDPR}    & Mean                 & \multicolumn{1}{c|}{\textbf{0.3185}} & \textbf{0.2242}   & \multicolumn{1}{c|}{\textbf{0.3188}} & \textbf{0.2266}   \\ \cline{2-6} 
                            & SD                   & \multicolumn{1}{c|}{\textbf{0.0087}} & 0.0144   & \multicolumn{1}{c|}{0.0170} & 0.0099   \\ \hline
\end{tabular}
\end{table}

\section{Link Prediction}\label{appendix_exp:link_pred}
A higher value of AUC indicates better utility.
From Tables~\ref{tab:auc_score1} and \ref{tab:auc_score}, PrivDPR consistently shows the highest AUC scores across all datasets. This indicates that it is the most effective algorithm for link prediction while maintaining node-level privacy. 
DPGGAN and DPGVAE, which both offer edge-level privacy, exhibit similar performance levels with $\epsilon=0.1$ and $\epsilon=3.2$. 
GAP has the lowest AUC scores compared to the other algorithms, indicating that it is the least effective for link prediction at both $\epsilon=0.1$ and $\epsilon=3.2$.
In the context of SD, a smaller value reflects greater stability in the algorithm. PrivDPR exhibits the lowest SD across all datasets under $\epsilon=0.1$, and also demonstrates the lowest SD for the Cora and Amazon datasets under $\epsilon=3.2$. This suggests that PrivDPR possesses strong stability. 
DPGVAE and DPGGAN exhibit relatively low SD in most datasets, signifying stable performance. Although GAP exhibits significant stability on the Chicago dataset under $\epsilon=3.2$, its overall performance is comparatively lower.

\section{Node Classification}\label{appendix_exp:node_class}
A higher Micro-F1 score indicates better utility. PrivDPR achieves the highest average Micro-F1 scores on the Cora and Citeseer datasets under both $\epsilon=0.1$ and $\epsilon=3.2$. This indicates PrivDPR is an effective algorithm for privacy-preserving node classification. DPGVAE follows with slightly poor Micro-F1 scores but shows stable performance. GAP scores slightly lower than PrivDPR and DPGVAE, and yet performs well especially with larger privacy budgets ($\epsilon=3.2$). DPGGAN obtains the lowest Micro-F1 scores, particularly struggling with smaller privacy budgets ($\epsilon=0.1$). In terms of SD, under $\epsilon=0.1$, it is observed that PrivDPR achieves the most stable result on Cora, while GAP achieves the most stable result on Citeseer. Additionally, both DPGVAE and PrivDPR demonstrate relatively stable performance across both datasets. Under $\epsilon=3.2$, GAP achieves the most stable result on Cora, while DPGVAE achieves the most stable result on Citeseer. Nevertheless, DPGGAN displays relatively stable performance across both Cora and Citeseer datasets.
\end{appendices}

\end{document}